\RequirePackage{tikz-cd}

\documentclass[sn-basic]{sn-jnl}

\usepackage{graphicx}
\usepackage{calligra,mathrsfs,mathtools}
\usepackage{bm,amsmath,amssymb}
\usepackage{enumitem}
\usepackage{amsthm}
\usepackage{extarrows}
\usepackage{faktor}
\usepackage[normalem]{ulem}
\usepackage{tikz-cd}
\newcommand{\com}[1]{\textcolor{red}{#1}}

\jyear{2022}%

\theoremstyle{thmstyleone}%
\theoremstyle{thmstyletwo}%

\theoremstyle{thmstylethree}%

\newtheorem{thm}{Theorem}[section]
\numberwithin{thm}{section}

\theoremstyle{definition}
\newtheorem{defn}[thm]{Definition}
\newtheorem{ex}[thm]{Example}
\newtheorem{rmk}[thm]{Remark}
\newtheorem{prop}[thm]{Proposition}
\newtheorem{corr}[thm]{Corollary}
\newtheorem{lem}[thm]{Lemma}
\numberwithin{equation}{section}

\begin{document}

\title[An algebraic characterization of self-generating chemical reaction networks]{An algebraic characterization of self-generating chemical reaction networks using semigroup models}


\author{\fnm{Dimitri} \sur{Loutchko}} \email{d.loutchko@edu.k.u-tokyo.ac.jp}

\affil{\orgdiv{Institute of Industrial Science}, \orgname{The University of Tokyo}, \orgaddress{\street{4-6-1, Komaba}, \city{Meguro-ku}, \postcode{153-8505}, \state{Tokyo}, \country{Japan}}}

\abstract{
The ability of a chemical reaction network to generate itself by catalyzed reactions from constantly present environmental food sources is considered a fundamental property in origin-of-life research.
Based on Kaufmann's autocatalytic sets, Hordijk and Steel have constructed the versatile formalism of catalytic reaction systems (CRS) to model and to analyze such self-generating networks, which they named reflexively autocatalytic and food generated (RAF).
Previously, it was established that the subsequent and simultaenous catalytic functions of the chemicals of a CRS give rise to an algebraic structure, termed a {\it semigroup model}. 
The semigroup model allows to naturally consider the function of any subset of chemicals on the whole CRS.
This gives rise to a generative dynamics by iteratively applying the function of a subset to the externally supplied food set.
The fixed point of this dynamics yields the maximal self-generating set of chemicals.
Moreover, the lattice of all functionally closed self-generating sets of chemicals is discussed and a structure theorem for this lattice is proven.
It is also shown that a CRS which contains self-generating sets of chemicals cannot be nilpotent and thus a useful link to the combinatorial theory of finite semigroups is established. 
The main technical tool introduced and utilized in this work is the representation of the semigroup elements as decorated rooted trees, allowing to translate the generation of chemicals from a given set of resources into the semigroup language.
}




\maketitle

\section{Introduction} \label{sec:intro}

Questions about the origin of life are as fascinating as they are difficult to even attempt to answer.
There are at least two schools of thought on how to approach such questions.
The first one is to construct minimal models involving concrete chemicals, best exemplified by the RNA world hypothesis formulated by \cite{gilbert1986}, \cite{joyce1989} and many others.
The great advantage of such concrete models is that they can be tested experimentally, going all the way back to the classical experiments by \cite{miller1953} and \cite{oro1961}.
However, there can never be certainty about any hypothesized model, and even the most convincing ones such as the RNA world hypothesis lack reliable data with regard to their first appearance, cf. \cite{joyce2002,penny2005}.
An alternative school of thought is focused on working out the minimal requirements which any sensible theory of the origin of life should satisfy.
Prominent proponents of this approach are \cite{Oparin1957}, \cite{Dyson1999}, \cite{Kauffman1986}, and many others.
However, already the formulation of a meaningful theoretical framework is challenging and there have been various attempts including $(M,R)$-systems by \cite{Rosen1958}, hypercycles by \cite{Eigen1971}, autopoetic systems by \cite{Varela1974}, chemotons by \cite{Ganti1975} and autocatalytic sets by \cite{Kauffman1986}.
A common feature that all frameworks have in common is the importance of autocatalysis and the occurrence of autocatalytic cycles as discussed in the review by \cite{Hordijk2018}.\\

The catalytic reaction system (CRS) formalism by \cite{Steel2000,Hordijk2004} is a versatile framework that, motivated by Kauffman's autocatalytic sets, captures the essence of several of the aforementioned approaches.
It has been used to compute thresholds for the occurrence of self-generating and self-sustaining motives in CRS based on the level of catalysis by \cite{Hordijk2010, Hordijk2011, Hordijk2012, Hordijk2015, Hordijk2017, Hordijk2018} and even for the analysis of the metabolic network of {\it E. Coli} by \cite{Sousa2015}.\\

In the companion article by \cite{Loutchko2019}, it has been shown that CRS have an algebraic structure that is generated by the simultaneous and subsequent function of chemicals acting as catalysts on the CRS.
It was then shown how a naturally defined discrete dynamics yields the maximal self-sustaining set of chemicals for any given CRS and a characterization of the lattice of functionally closed self-sustaining sets of chemicals was derived.
This article aims to achieve the same for self-generating sets of chemicals, which is a stricter notion than that of self-sustainment and requires more mathematical care.
In this regard, the main technical contribution of this article is to construct a representation of the semigroup elements as decorated rooted trees as they are naturally suited to deal with the {\it generation} of chemicals from a set of externally supplied chemicals.

\paragraph{Mathematical outline}

The construction of the semigroup models is based on the CRS formalism introduced by \cite{Hordijk2004,Hordijk2011}.
A CRS is given by the datum of a chemical reaction network, i.e. a finite set of chemicals $X$ together with a finite set of reactions $R$ where each reaction $r \in R$ is determined by the set of its reactants $\textrm{dom}(r) \subset X$ and products $\textrm{ran}(r) \subset X$.
Additionally, catalysis data is specified by a set $C \subset X \times R$ meaning that for each $(x,r) \in C$, the reaction $r$ is catalyzed by the chemical $x$, and a food set $F \subset X$ of constantly supplied chemicals is given.
A CRS is said to be RAF (reflexively autocatalytic and food-generated) if each chemical in the CRS can be generated from the food set $F$ by a series of catalyzed reactions.
A set of chemicals is said to be RAF if the CRS supported on it is RAF.
The notion of RAF formalizes self-generating reaction networks in the framework of CRS.
Details on CRS are given in Section \ref{sec:CRS}.\\

In Section \ref{sec:SemigroupBasics}, it is shown that the reactions and the catalytic functions of chemicals have the structure of a semigroup, which is additionally equipped with a partial order and an idempotent addition.
The semigroup operation corresponds to subsequent functionality whereas the addition corresponds to simultaneous application of functions.
More precisely, to each reaction $r \in R$ a function $\phi_r$ is assigned as the set-map $\phi_r: \mathfrak{X} \rightarrow \mathfrak{X}$ on the power set $\mathfrak{X}:=\mathcal{P}(X_F)$ of non-food chemicals $X_F = X \setminus F$.
The function $\phi_r$ gives the set of non-food products of $r$ if and only if the set of non-food reactants of $r$ is contained in its argument.
Such functions have the usual composition given by ${(\phi_r \circ \phi_{r'})(Y) = \phi_r (\phi_{r'}(Y))}$ and an idempotent addition given by ${(\phi_r + \phi_{r'})(Y) = \phi_r(Y) \cup \phi_{r'}(Y)}$ for all $Y \subset X_F$ and $r,r' \in R$.
They generate the semigroup model
\begin{align*}
 \mathcal{S}^R = \langle \phi_r \rangle _{r \in R}.
\end{align*}
To each of the chemicals $x \in X$, a function $\phi_x: \mathfrak{X} \rightarrow \mathfrak{X}$ is assigned by using the catalysis data:
\begin{equation*}
 \phi_x = \sum_{(x,r) \in C} \phi_r.
\end{equation*}
The functions of the chemicals generate the semigroup model 
\begin{equation*}
 \mathcal{S} = \langle \phi_x \rangle _{x \in X},
\end{equation*}
which is a subsemigroup of $\mathcal{S}^R$.
The objects $\mathcal{S}^R$ and $\mathcal{S}$ are semigroups with respect to both $+$ and $\circ$, hence they are called {\it semigroup models}.

The elements of the semigroup models are partially ordered via $\phi \leq \psi$ iff $\phi(Y) \subset \psi(Y)$ for all $Y \subset X_F$.
Lemma \ref{lem:elementaryProperties} states the the partial order on the semigroup models, the partial order on $\mathfrak{X}$, and the two operations $\circ$ and $+$ are all compatible.
A central notion is the function $\Phi_Y \in \mathcal{S}$ of a set of non-food chemicals $Y \subset X_F$ which is defined as the unique maximal element of the subsemigroup
\begin{align*}
 \mathcal{S}(Y) = \langle \phi_x \rangle _{x \in Y \cup F}
\end{align*}
of $\mathcal{S}$.
The function $\Phi_Y$ captures all catalytic functionality that can be exerted by $Y$ and the food set on all other chemicals of the CRS.\\

Section \ref{sec:prelim} provides more insight into the structure of the semigroup models.
The basis is the definition of a tree algebra $\mathfrak{T}(A)$ with a decorating algebra $(A,\circ,+)$ as follows:
The objects in $\mathfrak{T}(A)$ are rooted trees, whose edge labels are arbitrary elements in $A$.
The vertrex labels are determined by these edge labels:
All leaves are labelled by the multiplicatively neutral element $\textrm{id}$.
At each non-leaf vertex the labels of the outgoing edges are multiplied with the labels on the respective child vertex and the sum is taken over all the outgoing edges.
This is illustrated in Fig. \ref{fig:intro}A.
The addition of trees is performed by identifying their roots, with unchanged labels at the edges, as illustrated in Fig. \ref{fig:intro}B.
The multiplication of trees $T_1 \circ T_2$ is carried out by replacing all leaves of $T_1$ with copies of $T_2$.
Again, all edge labels are unchanged, as illustrated in Fig. \ref{fig:intro}C.

\begin{figure}[ht]
  \centering
  \includegraphics[scale=0.25]{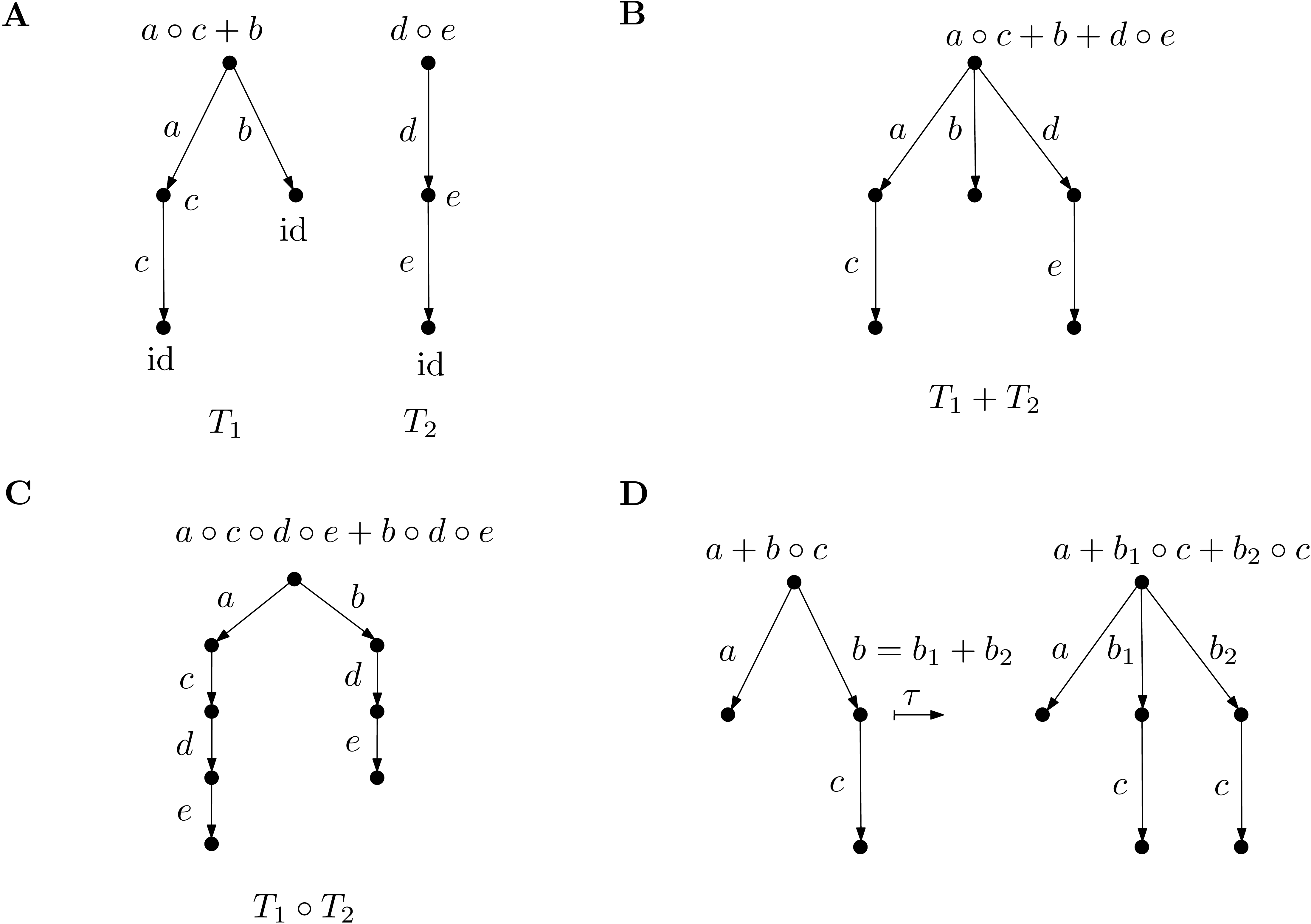}
  \caption{An illustration of the algebra of decorated rooted trees.
  {\bf A} The edge labels $a,b,c,d,e \in A$ determine the vertex labels of the two trees $T_1$ and $T_2$ recursively: the leaves are labelled by the multiplicatively neutral element and each vertex function is given by the summation of the labels over all outgoing edges, multiplied with the labels of the vertices at their heads.
  All edge and the resulting vertex labels are shown here, whereas in B,C and D only the labels of the edges and of the root are shown.
  {\bf B} Addition of the two trees $T_1$ and $T_2$: The roots of both trees are identified and the labels on all edges of both trees are retained.
  The vertex labels are determined as in A.
  The root label of $T_1 + T_2$ is equal to the sum of the root labels of $T_1$ and $T_2$.
  {\bf C} Multiplication of two trees $T_1$ and $T_2$: Each leaf of $T_1$ is replaced with a copy of $T_2$.
  Thereby, the edge labels from the original trees are retained, which yields the respective vertex labels.
  If the right distributivity of the operations $+$ and $\circ$ holds, then the root label of $T_1 \circ T_2$ is equal to the concatenation of the root labels of $T_1$ and $T_2$.
  {\bf D} The replacement of an edge with label $b = b_1 + b_2$ by two edges with labels $b_1$ and $b_2$.
  A copy of the child tree of the original edge is attached to each of the new edges.
  If the right distributivity of the operations $+$ and $\circ$ holds, then the root labels of both shown trees are equal.
  }
  \label{fig:intro}
\end{figure}

The tree algebras relevant for semigroup models have their edges labelled by the generating sets of the respective models, i.e. they are ${\mathfrak{T} := \mathfrak{T}(\{\phi_x\}_{x \in X} \cup \{0\})}$ and ${\mathfrak{T}^R := \mathfrak{T}(\{\phi_r\}_{r \in R} \cup \{0\})}$.
The main result of the section is Theorem \ref{thm:function}, which states that there is a commutative diagram of homomorphisms
\[\begin{tikzcd}
 \mathfrak{T}\arrow[r,twoheadrightarrow, "ev"]\arrow[d,hookrightarrow,"\tau"] & \mathcal{S} \cup \{ \textrm{id}\mid_{\mathfrak{X}} \} \arrow[d,hookrightarrow,"\iota"]\\
 \mathfrak{T}^R \arrow[r,twoheadrightarrow,"ev"]& \mathcal{S}^R \cup \{ \textrm{id}\mid_{\mathfrak{X}} \},
\end{tikzcd}\]
whereby the surjective evaluation map $ev$ sends the root label to the corresponding semigroup element and the map $\tau$ is defined based on the formula $\phi_x = \sum_{(x,r) \in C} \phi_r$.
More precisely, $\tau$ replaces an edge with the label $\phi_x$ by edges labeled by $\phi_r$ for each $(x,r) \in C$ and a copy of the child tree of the original edge is attached to each of the new edges, as illustrated in Fig. \ref{fig:intro}D.
A tree representing a semigroup element is a lift of the element via the evaluation homomorphism $ev$.
The algebraic reason for the existence of such representations is the interplay of the two operations $\circ$ and $+$ via the right distributivity $\phi \circ \chi + \psi \circ \chi = (\phi + \psi) \circ \chi$ for any $\phi, \psi, \chi \in \mathcal{S}^R$.

Loosely speaking, the trees in $\mathfrak{T}^R$ correspond to ''reaction mechanisms'', which proceed recursively from the leaves to the root such that a reaction labeling an edge occurs subsequently with the ''mechanism'' of its head vertex and such that all reactions labeling edges with same tail are carried out simultaneously.
Thus, it is natural to assume that a chemical $x \in X_F$ can be generated from the food set if there is a reaction mechanism for its generation, given by a tree $T \in \mathfrak{T}^R$.
This translates to $x \in ev(T)(\emptyset)$ in this setup.
And indeed, it is proven in Lemma \ref{lem:treesFproperty} that this property is equivalent to the standard definition of generation from the food set.\\

In Section \ref{sec:RAF}, it is shown how the representation of semigroup elements by decorated rooted trees can be used to describe CRS with the RAF property by the simple condition $\Phi_{X_F}(\emptyset) = X_F$ (Theorem \ref{thm:RAF}).
This implies that for a RAF set of chemicals $X'_F \subset X_F$, the property $X_F' \subset \Phi_{X_F'}(\emptyset)$ holds (Corollary \ref{corr:maxRAF}) and, moreover, that the equality $ X_F' = \Phi_{X_F'}(\emptyset)$ is a sufficient condition for $X'_F$ to be a RAF set of chemicals (Proposition \ref{prop:RAFsufficient}).
Then, a generative dynamics on $\mathfrak{X}$ is defined by $Y \mapsto \Phi_Y(\emptyset)$ and, as one of the main results, it is proven that the dynamics with initial condition given by $X_F$ leads to the maximal RAF set of chemicals.
Finally, new insights and conjectures gained from the semigroup approach to CRS with the RAF property are discussed.
It is shown that the generative dynamics with the initial condition given by a RAF set of chemicals $X_F'$ leads to a fixed point $X_F'^{*g}$, which contains $X_F'$.
If $X_F' \subsetneq X_F'^{*g}$ holds, then $X_F'$ is not stable because its own catalytic function will produce all chemicals in $X_F'^{*g}$ over time.
Therefore $X_F'^{*g}$ is termed the {\it functional closure} of $X_F'$.
A characterization of the lattice of all functionally closed RAF sets of a CRS is provided in Theorem \ref{thm:funcClosed}.
It is based on the reduced generative dynamics given by $Y \mapsto Y \cap \Phi_Y(\emptyset)$.
This dynamics always has a fixed point, denoted by $Y_0^{*rg}$ for the inital condition $Y_0$.
For each set $Y \subset X_F$, the set of fixed points
\begin{align*}
    \mathfrak{N}(Y) := \left\{ (Y \setminus \{y\})^{*rg} \textrm{ for }y \in Y \right\}
\end{align*}
is introduced and one recursively defines
\begin{align*}
    \mathfrak{N}^0 &:= \{ X_F^{*g} \} \\
    \mathfrak{N}^{i+1} &:= \bigcup_{Y \in  \mathfrak{N}^i} \mathfrak{N}(Y) \textrm{ for all }i \in \mathbb{Z}_{\geq 0}.
\end{align*}
The statement of the Theorem \ref{thm:funcClosed} is that the lattice of functionally closed RAF sets of chemicals is given by 
\begin{equation*}
    \mathfrak{N} := \bigcup_{i=0}^{\mid X_F \mid} \mathfrak{N}^i.
\end{equation*}

In the concluding Section \ref{sec:discussion}, the importance of the representations of semigroup elements by decorated rooted trees is discussed, and the biochemical significance of functionally closed RAF sets of chemicals is illustrated.
For example, one would expect chemicals which are uniquely contained in a minimal functionally closed RAF set of chemicals to be involved solely in the functionality of the respective RAF set, whereas chemicals that have multiple minimal functionally closed RAF sets of chemicals containing them are more likely to be involved in communication and interaction between the respective RAF sets.
This can potentially carry information on the evolutionary role of the respective chemicals.
This is an illustration of how the semigroup models can be used to discover new concepts in CRS theory.
In future work, such concepts will be applied to CRS corresponding to real biological systems.

\section{Semigroup models} \label{sec:semigroup models}

The construction of semigroup models and their elementary properties are provided in Section \ref{sec:SemigroupBasics}.
They are based on the catalytic reaction system (CRS) formalism, which is introduced in Section \ref{sec:CRS}.
This is a condensed version of the Sections 2. and 3. from the introductory companion article by \cite{Loutchko2019}.
Only the RAF property (Definition \ref{def:RAF}) and the extended semigroup model $\mathcal{S}^R$ (Definition \ref{def:semigroupModel}) are newly introduced here.

\subsection{The CRS formalism} \label{sec:CRS}

The introduction of the catalytic reaction system (CRS) formalism and of the reflexievly-autocatalytic and food generated (RAF) property are based on the work of \cite{Hordijk2004}.

\noindent The notion of CRS is designed to capture the catalytic functionality within a given chemical reaction network. 
It does not take into account detailed kinetic or thermodynamic information.

\begin{defn} \label{def:CRS}
A {\it catalytic reaction system} (CRS) is a tuple $(X,R,C,F)$ where $X$ is a finite discrete set of chemicals, $R$ is a finite set of reactions, $C \subset X \times R$ is the catalysis data for the reactions $R$ and $F \subset X$ is the constantly present food set.

\noindent Each reaction $r \in R$ is given by a pair $(\textrm{dom}(r),\textrm{ran}(r))$ of mutually disjoint subsets of $X$, called the {\it domain} and the {\it range} of $r$.
The elements of $\textrm{dom}(r)$ are called the reactants and the elements of $\textrm{ran}(r)$ are the products of $r$.
For a pair $(x,r) \in C$, the reaction $r$ is said to be {\it catalyzed} by $x$ and $x$ is said to be a {\it catalyst} of $r$.
The food set $F$ is required to satisfy the following closure property:
\begin{enumerate} [label=(C),leftmargin=1cm]
    \item All reactions $r \in R$ with a catalyst in $F$ must involve chemicals outside of $F$ as reactants, i.e. they must satisfy $\text{dom}(r) \cap (X \setminus F) \neq \emptyset$.  \label{cond:C}
\end{enumerate}
If $X=F$, the CRS is said to be {\it trivial}.
\end{defn}

\begin{ex}
Fig. \ref{fig:example1} shows a representation of a CRS as a directed bipartite graph.
This representation is used throughout this article.
The chemicals are represented by solid vertices and the reactions $r = (\textrm{dom}(r),\textrm{ran}(r))$ are represented by circles.
For each reaction, there are directed edges from each chemical in $\textrm{dom}(r)$ to the reaction vertex and from the reaction vertex to each chemical in $\textrm{ran}(r)$.
The catalysis data $(x,r) \in C$ is indicated by a dashed directed edge from the chemical $x$ to the reaction $r$.
The food set is indicated by a circle around the food chemicals.
\begin{figure}[htb]
  \centering
  \includegraphics[scale=0.25]{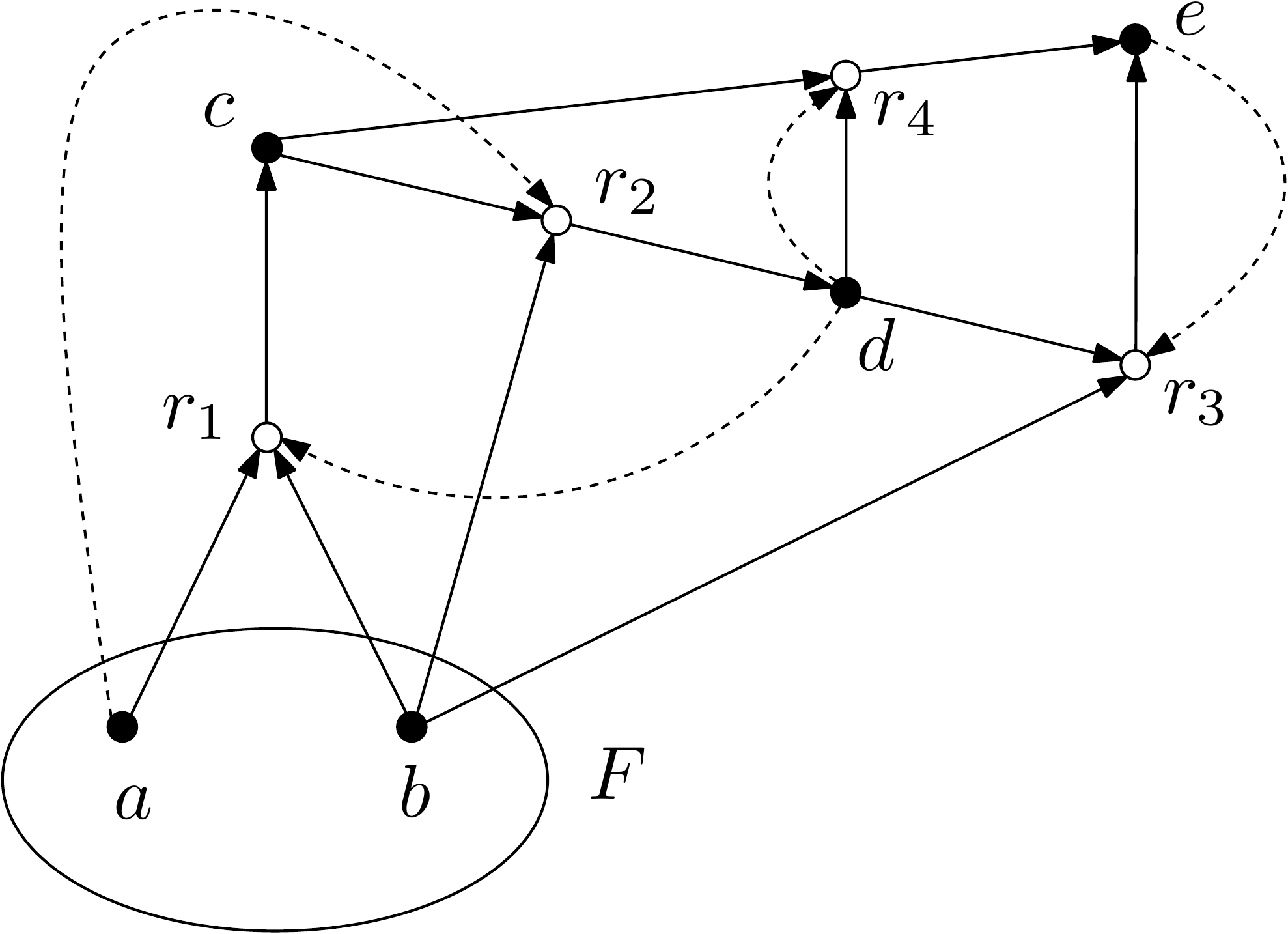}
  \caption[Example of a catalytic reaction system]{Example of a graphical representation of a CRS. 
  The CRS consists of five chemicals $X=\{a,b,c,d,e\}$ and four reactions $a+b\rightarrow c$, $c+b \rightarrow d$, $b+d \rightarrow e$ and $c+d \rightarrow e$, which are catalyzed by $d$, $a$, $e$ and $d$, respectively.
  The food set is given by $F = \{a,b\}$.}
  \label{fig:example1}
\end{figure}
\end{ex}

The projection
\begin{align*}
    \pi_R :C \rightarrow R
\end{align*}
yields the set of all catalyzed reactions of the CRS as $\pi_R(C)$.

\noindent The {\it support} of a reaction $r$ is defined as
\begin{align*}
    \textrm{supp}(r) := \textrm{dom}(r) \cup \textrm{ran}(r)
\end{align*}
and the notions of domain, range, and support extend to sets of reactions $R' \subset R$ via
\begin{align*}
    \textrm{dom}(R') = \bigcup_{r \in R'} \textrm{dom}(r)
\end{align*}
with the analogous definitions for $\textrm{ran}(R')$ and $\textrm{supp}(R')$.

From now on, a CRS $(X,R,C,F)$ will be fixed.
When referring to any of the four sets $X,R,C$ or $F$, it is implicitly assumed that they are part of the full data of the CRS.
It will be convenient to abbreviate the non-food chemicals as
\begin{align*}
    X_F := X \setminus F
\end{align*}
and to make the same definition for any subset $X'$ of $X$ containing $F$, i.e. $X'_F := X' \setminus F$.
Moreover, given a set $X'_F \subset X_F$, the symbol $X'$ will denote the set $X'_F \cup F \subset X$. 

\begin{defn} \label{def:subCRS}
For a set $X'_F \subset X_F$ of non-food chemicals, define the restrictions of $R$ and $C$ as
\begin{align*}
    R\mid_{X'} &= \{ r \in R \text{ such that supp}(r) \subset X'\}, \\
    C\mid_{X'} &= C \cap (X' \times R\mid_{X'}).
\end{align*}
The tuple $(X',R\mid_{X'},C\mid_{X'},F)$ is called the {\it subCRS generated by $X'_F$}.
\end{defn}

In the article by \cite{Loutchko2019}, a broader notion of subCRS is introduced.
This notion is, however, not needed in this work as the focus will be exclusively on subCRS generated by sets of non-food chemicals.
Note that the subCRS according to Definition \ref{def:subCRS} is always closed in the terminology used by \cite{Hordijk2017}, i.e. all reactions of the full CRS with support on $X'$ are actually contained in the respective subCRS.

Now the central notions of food generated CRS and reflexively autocatalytic and food generated (RAF) CRS are introduced following \cite{Hordijk2004,Hordijk2017}.
However, the definitions given by \cite{Hordijk2004,Hordijk2017} are centered around the set of reactions $R$, whereas the definitions given here involve the whole CRS. 
In Remark \ref{rmk:RAFliterature}, the relation to the definitions used in this work is discussed.
The F property formalizes the idea that all chemicals of the CRS can be {\it generated} from the food set.
The RAF property means that the generation from the food set can be achieved with catalyzed reactions only.

\begin{defn} \label{def:Fnetwork}
A CRS $(X,R,C,F)$ has the {\it food generation property} (F property) if each $x \in X_F$ is generated by some sequence of reactions from $F$, i.e. if the following condition is satisfied for each $x \in X_F$:
\begin{enumerate} [label=(F),leftmargin=1cm]
    \item There exist sets of reactions $R_1,...,R_n \subset R$ with the following properties: \label{cond:F}
\end{enumerate} 
\begin{enumerate} [label=(F\arabic*),leftmargin=2cm]
    \item $\textrm{dom}(R_1) \subset F$. \label{cond:F1}
    \item $\textrm{dom}(R_{i+1}) \subset \bigcup_{j=1}^i \textrm{ran}(R_j) \cup F$ for all $1 \leq i \leq n-1$. \label{cond:F2}
    \item $x \in \textrm{ran}(R_n)$. \label{cond:F3}
\end{enumerate}
\end{defn}

\begin{defn} \label{def:RAF}
A CRS $(X,R,C,F)$ is {\it refelxively autocatalytic and food generated} (RAF) if it is has the F property and if for each chemical $x \in X_F$, the sets of reactions $R_1,\dots R_n \subset R$ featured in the condition \ref{cond:F} can be chosen to be subsets of $\pi_R(C)$. 
In other words, the reactions in $R_1,\dots R_n$ are all required to be catalyzed.
\end{defn}

\begin{rmk} \label{rmk:selfSustaining}
The notion of {\it self-generation} is stronger than the one of {\it self-sustainment}.
Self-sustaining CRS are treated within the semigroup formalism by \cite{Loutchko2019}.
Self-sustainment requires the CRS to have a catalyzed set of reactions $R' \subset \pi_R(C)$ such that $\textrm{dom}(R') \subset \textrm{ran}(R') \cup F$ and $X_F \subset \textrm{ran}(R')$.
The RAF condition is stronger than this, because one can set $R'_x := \cup_{i=1}^n R_n$ for the reactions featured in condition \ref{cond:F} and $R' = \cup_{x \in X_F} R'_x$ will satisfy the requirement for self-sustainment.
On the contrary, there are self-sustaining CRS which are not self-generating.
\end{rmk}

\noindent The definition of the RAF property descends to sets of non-food chemicals $X'_F \subset X_F$ based on the Definition \ref{def:subCRS}.

\begin{defn} \label{def:RAFsetofchemicals}
    A set of chemicals $X'_F \subset X_F$ is said to be a {\it RAF set of chemicals} if the subCRS $(X',R\mid_{X'},C\mid_{X'},F)$ generated by it is RAF.
\end{defn}

\begin{ex}
The CRS in Fig. \ref{fig:example1} is RAF and thus $X_F = \{c,d,e\}$ is a RAF set of chemicals.
Moreover, there is a RAF subset of chemicals consisting of $X_F' = \{c,d\}$, because $d$ catalyzes the formation of $c$ from the food set and $c$ reacts with the food set to form $d$, which is catalyzed by the food set.
\end{ex}

\begin{rmk}[Relation to the notion of RAF commonly used in the literature] \label{rmk:RAFliterature}
The Definition \ref{def:Fnetwork} of the F property given here coincides {\it verbatim} with the one commonly used in the CRS literature.
The Definition \ref{def:RAF} of the RAF property is equivalent to the definitions of a closed\footnote{In \cite{Hordijk2017}, a subset $R' \subset R$ is called a closed RAF set of reactions if it is RAF and if, in addition, all reactions with a catalyst and support in $\textrm{supp}(R')$ are elements of $R'$.} RAF set of reactions given by \cite{Hordijk2004,Hordijk2017} modulo the inclusion of uncatalyzed reactions in the set of reactions $R$ in the definition given here.
\cite{Hordijk2004,Hordijk2017} define the RAF property for subsets of $R$ as follows:
A subset $R' \subset R$ is a RAF set of reactions if it has the F property\footnote{The F property for a set $R' \subset R$ means that each chemical in $\textrm{supp}(R') \cap X_F$ satisfies the condition \ref{cond:F} with the $R_i$ satisfying $R_i \subset R'$.} and if each reaction $r \in R'$ is catalyzed by a chemical $x \in \textrm{supp}(R') \cup F$.
Thus, a subCRS $(X',R\mid_{X'},C\mid_{X'},F)$ corresponds the RAF set of reactions $\pi_{R'}(C')$ and, {\it vice versa}, a closed RAF set of reactions $R'$ corresponds to the CRS $(X',R\mid_{X'},C\mid_{X'},F)$ with $X' := \textrm{supp}(R')$.

One can easily lift the restriction of the RAF sets of reactions being closed by defining subCRS with sets of chemicals $X'$ to allow for arbitrary sets of reactions $R' \subset R\mid_{X'}$.
This construction is given by \cite{Loutchko2019}.
\end{rmk}

\subsection{The semigroup model of a CRS} \label{sec:SemigroupBasics}

The chemical reactions of a CRS have a natural algebraic structure given by the simultaneous and subsequent occurrence of reactions, as well as combinations thereof.
Making this mathematically precise leads to the notion of an extended semigroup model $\mathcal{S}^R$ of a CRS.
The function of a chemical is defined by the simultaneous occurrence of all the reactions it catalyzes.
All combinations of subsequent and simultaneous functions of chemicals give rise to the semigroup model $\mathcal{S}$ of a CRS.
The construction of the semigroup models is motivated by the work of \cite{Rhodes2010} in spirit, but technically the objects constructed here differ significantly, cf. \cite{Loutchko2019}, Remark 3.4.

Throughout this section, let $(X,R,C,F)$ be a CRS.
The state of the CRS is defined by the presence or absence of each of the non-food chemicals, i.e. by a subset $Y \subset X_F$.
Therefore, the state space $\mathfrak{X}$ of the CRS is the power set
\begin{align*}
    \mathfrak{X} := \mathcal{P}(X_F) = \{0,1\}^{X_F}.
\end{align*}
A reaction $r \in R$ acts on the state space via its {\it function}
\begin{equation*}
 \phi_r : \mathfrak{X} \rightarrow \mathfrak{X}
\end{equation*}
given by
\begin{align} \label{eq:phi_r}
  \phi_r(Y) =\begin{cases}
    \text{ran}(r) \cap X_F & \text{if $\text{dom}(r) \subset Y \cup F$}\\
	\emptyset & \text{else}
  \end{cases}
\end{align}
for all $Y \subset X_F$.
Two maps $\phi, \psi: \mathfrak{X} \rightarrow \mathfrak{X}$ can be composed via the addition $+$, which is defined as
\begin{equation} \label{eq:addition}
  (\phi + \psi)(Y) = \phi(Y) \cup \psi(Y)
\end{equation} for all $Y \subset X_F$.
This operation is associative, commutative and idempotent.
Moreover, the multiplication $\circ$ is given by the usual composition of maps
\begin{equation} \label{eq:multiplication}
(\phi \circ \psi)(Y) := \phi (\psi(Y))
\end{equation}
for all $Y \subset X_F$.

\noindent Finally, the {\it function} $\phi_x: \mathfrak{X} \rightarrow \mathfrak{X}$ of a chemical $x \in X$ is defined as the sum over all reactions catalyzed by it via
\begin{equation} \label{eq:phi_x}
\phi_x = \sum_{(x,r) \in C} \phi_r.
\end{equation}

Recall that the {\it full transformation semigroup} $\mathcal{T}(A)$ of a finite discrete set $A$ is the set of all maps $\{f:A \rightarrow A\}$, where the semigroup operation $\circ$ is the composition of maps.
The semigroup model of a CRS is a subsemigroup of $\mathcal{T}(\mathfrak{X})$ and is defined as follows.

\begin{defn}  \label{def:semigroupModel}

The {\it semigroup model} $\mathcal{S}$ of a CRS is a subsemigroup of $\mathcal{T}(\mathfrak{X})$ generated by the functions $\{\phi_x\}_{x \in X}$ through the operations of addition and composition, i.e. $\mathcal{S}$ is the smallest subsemigroup of the full transformation semigroup $\mathcal{T}(\mathfrak{X})$ closed under $\circ$ and $+$ that contains $\{\phi_x\}_{x \in X}$ and the zero function, given by $0(Y) = \emptyset$ for all $Y \subset X_F$.
It is denoted by
\begin{equation*}
\mathcal{S} = \langle \phi_x \rangle_{x \in X}.
\end{equation*}
Analogously, the {\it extended semigroup model} of the CRS is generated by the functions $\phi_r$ of all reactions $r \in R$.
This model is denoted as 
\begin{equation*}
\mathcal{S}^R = \langle \phi_r \rangle_{r \in R}.
\end{equation*}
\end{defn}

As subsemigroups of $\mathcal{T}(\mathfrak{X})$, the semigroups $\mathcal{S}$ and $\mathcal{S}^R$ are {\it finite}.
The objects $\mathcal{S}$ and $\mathcal{S}^R$ are called semigroup models, because they are semigroups with respect to both operations $\circ$ and $+$.
The correct description in terms of universal algebra is, however, an algebra of type $(2,2,0)$, cf. \cite{Almeida1995}.
The semigroup model $\mathcal{S}^R$ contains $\mathcal{S}$ as a subalgebra of type $(2,2,0)$ and this will be expressed by saying that $\mathcal{S}$ is a {\it subsemigroup model} of $\mathcal{S}^R$.

\begin{rmk} \label{rmk:PO}
In addition to the two algebraic operations, there is a natural partial order on $\mathcal{S}^R$ and $\mathcal{S}$, given by $\phi \leq \psi \Leftrightarrow \phi(Y) \subset \psi(Y) \text{ for all $Y \subset X_F$}$ for $\phi, \psi \in \mathcal{S}^R$.
\end{rmk}

\noindent There is an important subsemigroup of $\mathcal{S}$ generated by the functions of chemicals in a given set $X_F' \subset X_F$ together with the food set:
\begin{defn} \label{def:function}
For a subset $X'_F$ of $X_F$, the semigroup model $\mathcal{S}(X_F') < \mathcal{S}$ generated by the functions of $X_F'$ is defined as
\begin{equation*}
\mathcal{S}(X'_F) = \langle \phi_x \rangle_{x \in X'}
\end{equation*}
and the {\it function} $\Phi_{X_F'}$ of the set $X_F'$ is given by
\begin{equation*} \label{eq:PhiY}
\Phi_{X_F'} = \sum_{\phi \in \mathcal{S}(X_F')} \phi.
\end{equation*}
\end{defn}

The semigroup models satisfy the following elementary properties.
These properties follow directly from the definitions.
However, if necessary, the proofs for the respective statements on $\mathcal{S}$ can be found in \cite{Loutchko2019}, Section 3.2., and the proofs for $\mathcal{S}^R$ are analogous.

\begin{lem}[Elementary properties of semigroup models] \label{lem:elementaryProperties}
~\newline
\begin{enumerate} [label=(S\arabic*),leftmargin=1cm]
    \item All elements $\phi \in \mathcal{S}^R$ respect the partial order on $\mathfrak{X}$ given by inclusion of sets, i.e.
    \begin{equation*}
    Z \subset Y \subset X_F \implies \phi(Z) \subset \phi(Y).
    \end{equation*}
    \label{eq:generators}
    \item The partial order is compatible with addition and multiplication, i.e. for any $\phi, \phi', \psi, \psi' \in \mathcal{S}^R$ the following relations hold
    \begin{align}
    \label{eq:multCompatibility}
    \phi \leq \psi \textrm{ and } \phi' \leq \psi' &\Rightarrow \phi \circ \phi' \leq \psi \circ \psi',\\
    \label{eq:addCompatibility}
    \phi \leq \psi \textrm{ and } \phi' \leq \psi' &\Rightarrow \phi + \phi' \leq \psi + \psi'.
    \end{align} \label{lemma:PO}
    \item Any $\phi, \psi \in \mathcal{S}$ satisfy 
    \begin{equation*}
    \phi \leq \phi + \psi.
    \end{equation*} \label{eq:order1}
    \item Any $\phi, \phi', \psi \in \mathcal{S}$ such that $\phi \leq \psi$ and $\phi' \leq \psi$ satisfy
    \begin{equation*} 
    \phi + \phi' \leq \psi.
    \end{equation*} \label{eq:order2}
    \item The operations $\circ$ and $+$ on $\mathcal{S}^R$ have the following distributivity properties:
    \begin{align}
    \label{eq:distrRight}
    \phi \circ \chi + \psi \circ \chi &= (\phi + \psi) \circ \chi, \\ 
    \label{eq:distrLeft}
    \chi \circ  \phi+   \chi \circ  \psi &\leq  \chi \circ (\phi + \psi)
    \end{align}
    hold for any $\phi, \psi, \chi \in \mathcal{S}^R$.
    \label{lemma:distr}
    \item
    The right distributivity in Equation (\ref{eq:distrRight}) holds more generally for arbitrary elements $\phi, \psi, \chi \in \mathcal{T}(\mathfrak{X})$.
    \label{prop:distr}
    \item $\Phi_{X'_F}$ is the unique maximal element of $\mathcal{S}(X'_F)$.
    In particular, $\mathcal{S}$ has a unique maximal element, given by $\Phi_{X_F}$. \label{prop:function}
    \item The functions of sets $X''_F \subset X'_F \subset X_F$ satisfy 
    \begin{equation*}
    \Phi_{X''_F} \leq \Phi_{X'_F}.
    \end{equation*} \label{rmk:POtranssitivity}
\end{enumerate} 
\end{lem}

\begin{rmk} \label{rmk:subsemigroup}
Any subCRS $(X',R\mid_{X'},C\mid_{X'},F)$ generated by the set of chemicals $X_F'$ has a semigroup model given by Definition \ref{def:semigroupModel}, which will be denoted by $\mathcal{S}'(X',R\mid_{X'})$.
It is a subsemigroup of the full transformation semigroup $\mathcal{T}(\mathcal{P}(X'_F))$ on the power set of $X'_F$.
Any element $\phi \in \mathcal{S}'(X',R\mid_{X'})$ can be extended to a function $ext(\phi) \in \mathcal{T}(\mathfrak{X})$ via
\begin{equation*}
    ext(\phi)(Y) = \phi(Y \cap X'_F)
\end{equation*}
for $Y \subset X_F$.
By definition, the generators $\{\phi'_x\}_{x \in X'} \subset \mathcal{T}(\mathcal{P}(X'_F))$ of $\mathcal{S}'(X',R\mid_{X'})$ and the generators $\{\phi_x\}_{x \in X'} \subset \mathcal{T}(\mathfrak{X})$ of $\mathcal{S}(X'_F)$ satisfy ${ext(\phi'_x) \leq \phi_x}$ for all $x \in X'$.
Together with the property \ref{lemma:PO} this yields the inequality
\begin{align} \label{eq:maxInequality}
    ext(\Phi'_{X'_F}) \leq \Phi_{X'_F}
\end{align}
for the maximal functions $\Phi'_{X'_F}$ and $\Phi_{X'_F}$ of $\mathcal{S}'(X',R\mid_{X'})$ and $\mathcal{S}(X'_F)$.
\end{rmk}

\noindent This finishes the summary of the elementary properties of the semigroup models.
In the next section, a representation of the semigroup elements, which is well suited to deal with the condition \ref{cond:F} in food generated CRS, is constructed.

\section{Semigroup Models as Decorated Rooted Trees} \label{sec:prelim}

This section is dedicated to the construction of a representation of elements of $\mathcal{S}$ as decorated rooted trees.
It forms the technical basis for the proofs in the next section.
Albeit the main idea of this section is rather straightforward, the verification of all the claimed properties requires some care.
Therefore, the reader might prefer to skip this section up until Theorem \ref{thm:function} during the first reading.

The general idea developed in this section is as follows:
The edges of the rooted trees are labeled by functions in a subset of the full transformation semigroup $\mathcal{T}(\mathfrak{X})$. 
Each vertex is labelled by the sum of the functions on the outgoing edges multiplied with the functions of the respective head vertices (Definition \ref{def:tree}, see Fig. \ref{fig:tree} for an illustration).
Moreover, there are operations of addition and multiplication (Definition \ref{def:trees} and Fig. \ref{fig:tree_operations}) on the set of decorated rooted trees that are compatible with the addition and multiplication of the semigroup elements on the root (Lemma \ref{lem:tree_homo}).
The addition of two trees is performed by identifying their roots, and the multiplication is given by replacing the leaves of first tree with copies the second tree.
Finally, to establish a relation to the semigroup models $\mathcal{S}^R$ and $\mathcal{S}$, the edge labels are chosen from the generating sets $\{\phi_r\}_{r \in R} \cup \{0\}$ and $\{\phi_x\}_{x \in X} \cup \{0\}$, respectively. 
This idea is also sketched in the mathematical outline in the introductory Section \ref{sec:intro}.
The main Theorem \ref{thm:function} of this section establishes that both classes of decorated rooted trees are compatible with the algebraic structure of the semigroup models.
The merit of this construction is that the F and RAF properties of a CRS can be reformulated in terms of decorated rooted trees and then directly cast into the language of semigroup models (Lemma \ref{lem:treesFproperty}).\\

The following notations and conventions with regard to rooted trees will be used.
Let $T=(V,E,t)$ be a rooted tree with vertex set $V$, edge set $E \subset V \times V$ and root $t \in V$.
Edges $(v,w) \in E$ are directed from $v$ to $w$.
Here, $v$ is called the {\it tail} of $e$ and $w$ is its {\it head}.
For each vertex $v \in V$, let $\textrm{ch}(v) \subset V$ denote the set of {\it children} of $V$, which is defined as $\textrm{ch}(v) := \{w \in V \textrm{ such that }(v,w) \in E \}$.
Also, denote by $T_v$ the subtree of $T$ rooted at the vertex $v$.
The {\it level} $\textrm{lv}(v)$ of a vertex $v$ is the length of the path from the root to $v$ and $\textrm{lv}_n(T) \subset V$ denotes the set of all vertices of a given level $n$.
Moreover, the non-standard notation $\textrm{elv}_n(T) \subset E$ denotes the set of edges of {\it level} $n$, which are all the edges whose head vertex has level $n$.
The notation $\mathrm{ht}(T)$ denotes the {\it height} of the tree, i.e. the length of the longest path from the root to a leaf.
Finally, $\mathrm{lf}(T)$ is the set of all leaves of $T$, which is given by $\mathrm{lf}(T) := \{v \in V \textrm{ such that }\textrm{ch}(v) = \emptyset \}$.
An edge $(v,w) \in E$ is said to be {\it terminal} if the vertex $w$ is a leaf.

\begin{defn} \label{def:tree}
For any subset $A \subset \mathcal{T}(\mathfrak{X})$ of the full transformation semigroup $\mathcal{T}(\mathfrak{X})$, an {\it $A$-decorated rooted tree} $T = (A,V,E,t,\omega_V,\omega_E)$ is a finite rooted tree with vertex set $V$, edge set $E$, a root $t \in V$ and two maps
\begin{align*}
    \omega_V &: V \longrightarrow \mathcal{T}(\mathfrak{X}) \\
    \omega_E &: E \longrightarrow A,
\end{align*}
\noindent where $\omega_V$ is recursively given by
\begin{align} \label{eq:omega}
     \omega_V(v) = \begin{cases}
               \textrm{id}\mid_{\mathfrak{X}}  &\textrm{if }v \in \textrm{lf}(T) \\
               \sum_{w \in \textrm{ch}(v)} \omega_E((v,w)) \circ \omega_V(w) \qquad &\textrm{else}.
            \end{cases}
\end{align}
\noindent The addition and multiplication in the definition of $\omega_V$ takes place inside $\mathcal{T}(\mathfrak{X})$ as previously defined (cf. Equations (\ref{eq:addition}) and (\ref{eq:multiplication})).
Fig. \ref{fig:tree} illustrates this construction.
\end{defn}

\begin{figure}[ht]
  \centering
  \includegraphics[scale=0.25]{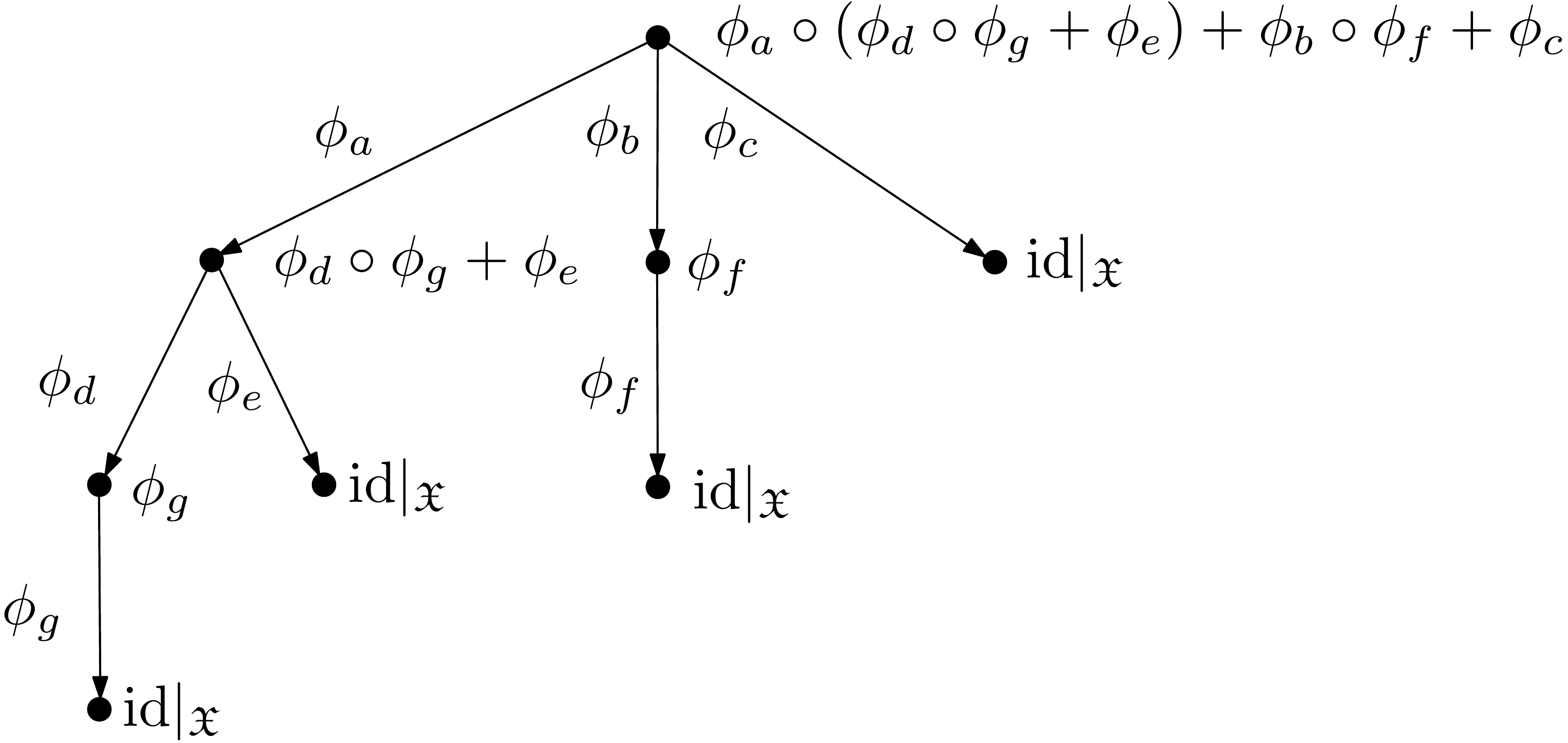}
  \caption{
  Example of a decorated rooted tree with decorations from the generating set $A = \{ \phi_x \}_{x \in X_F} \cup \{0\}$ of a semigroup model $\mathcal{S} = \langle \phi_x \rangle _{x \in X_F} \subset  \mathcal{T}(\mathfrak{X})$.
  The labels of the edges determine the labels on the vertices recursively:
  At each vertex, a sum over the labels of its children, multiplied by the labels on the respective connecting edges, is taken.
  The edges are labeled to the left of the respective edge and the resulting labels of the vertices are on the right of the respective vertex.
  The root is labelled by the function $\phi_a \circ (\phi_d \circ \phi_g + \phi_e) + \phi_b \circ \phi_f + \phi_c \in \mathcal{S}$.
  }
  \label{fig:tree}
\end{figure}

Decorated rooted trees will be referred to as trees.
For the set of edge labels $A \subset \mathcal{T}(\mathfrak{X})$, denote the set of all $A$-decorated trees by $\mathfrak{T}(A)$.
Also denote the set of all $A$-decorated trees of height $n$ by $\mathfrak{T}(A)^n$ and of height at most $n$ by $\mathfrak{T}(A)^{\leq n}$.
A subtree is defined as follows.
\begin{defn} \label{def:rooted_subtree}
A decorated rooted tree $T' = (A,V',E',t',\omega'_{V'},\omega'_{E'}) \in  \mathfrak{T}(A)$ is a {\it subtree} of $T = (A,V,E,t,\omega_V,\omega_E) \in  \mathfrak{T}(A)$ if there exists an injective map of rooted trees 
\begin{equation*}
    g: (V',E',t') \rightarrow (V,E,t),
\end{equation*}
which respects the labels on the edges, i.e.
\begin{equation*}
    \omega'_{E'}(e') = \omega_E(g(e'))
\end{equation*}
holds for all $e' \in E'$.
\end{defn}

The set $\mathfrak{T}(A)$ is equipped with two operations:
Loosely speaking, given two trees $T_1, T_2 \in \mathfrak{T}(A)$, their sum is defined by identifying the roots of $T_1$ and $T_2$ and their product by replacing each leaf of $T_1$ with a copy of $T_2$.

\begin{defn} \label{def:trees}
Let $T_1,T_2 \in \mathfrak{T}(A)$ be two $A$-decorated rooted trees given by the data $T_1 = (A,V_1,E_1,t_1,\omega_{V1},\omega_{E1})$ and $T_2 = (A,V_2,E_2,t_2,\omega_{V2},\omega_{E2})$.
Define the tree $T^+ := T_1 + T_2$ with data $T^+ = (A,V^+,E^+,t^+,\omega_V^+,\omega_E^+)$ by identifying the roots of the two trees, i.e. by
\begin{equation*}
    V^+ := \faktor{V_1 \sqcup V_2 }{t_1 \sim t_2}.
\end{equation*}
\noindent There is a canonical map
\begin{equation*}
   \epsilon^+: (V_1 \times V_1) \sqcup (V_2 \times V_2) \hookrightarrow (V_1 \sqcup V_2) \times (V_1 \sqcup V_2) \twoheadrightarrow V^+ \times V^+.
\end{equation*}
\noindent The edge set $E^+$ is defined as
\begin{equation*}
    E^+ := \epsilon^+(E_1 \sqcup E_2)
\end{equation*}
with the decoration map
 \begin{align} \label{eq:E+}
     \omega_E^+(e) = \begin{cases}
               \omega_{E1}((\epsilon^+)^{-1}(e))  &\textrm{if }(\epsilon^+)^{-1}(e) \in E_1\\
               \omega_{E2}((\epsilon^+)^{-1}(e))  &\textrm{if }(\epsilon^+)^{-1}(e) \in E_2.
            \end{cases}
 \end{align}
\noindent Because the restriction of $\epsilon^+$ to $E_1 \sqcup E_2$ is one-to-one, this map is well-defined.
The map $\omega_V^+$ is given by the relation (\ref{eq:omega}) with $E^+$ and $\omega_E^+$ instead of $E$ and $\omega_E$.
The construction is illustrated in Fig. \ref{fig:tree_operations}A.

Moreover, define the tree $T^{\circ} := T_1 \circ T_2$ with data $T^{\circ} = (A,V^{\circ},E^{\circ},t^{\circ},\omega_V^{\circ},\omega_E^{\circ})$ by replacing each leaf of $T_1$ with a copy of $T_2$.
The data on $T^{\circ}$ is given as follows.
\begin{equation*}
    V^{\circ} := \faktor{V_1 \sqcup \coprod_{l \in \textrm{lf}(T_1)} V_2 }{\sim },
\end{equation*}
\noindent where the equivalence relation $\sim$ relates each leaf $l \in \textrm{lf}(T_1) \subset V_1$ with the root $t_2 \in V_2$ of the respective copy of $V_2$ indexed by $l$.
Again, there is a canonical map
\begin{equation*}
   \epsilon^{\circ}: (V_1 \times V_1) \sqcup \coprod_{l \in \textrm{lf}(T_1)} (V_2 \times V_2) \rightarrow V^{\circ} \times V^{\circ}
\end{equation*}
\noindent and the edge set is defined as $E^{\circ} := \epsilon^{\circ}(E_1  \sqcup \coprod_{l \in \textrm{lf}(T_1)} E_2)$.
The restriction of $\epsilon^{\circ}$ to $E_1  \sqcup \coprod_{l \in \textrm{lf}(T_1)} E_2$ is one-to-one such that $\omega_E^{\circ}$ is defined analogously to $\omega_E^+$ as in (\ref{eq:E+}).
The map $\omega_V^{\circ}$ is defined by the relation (\ref{eq:omega}) using $E^{\circ}$ and $\omega_E^{\circ}$ instead of $E$ and $\omega_E$.
This construction is illustrated in Fig. \ref{fig:tree_operations}B.

The set $\mathfrak{T}(A)$, together with the two operations $\circ$ and $+$, is referred to as the {\it tree algebra} $\mathfrak{T}(A)$.
\end{defn}

\begin{figure}[ht]
  \centering
  \includegraphics[scale=0.25]{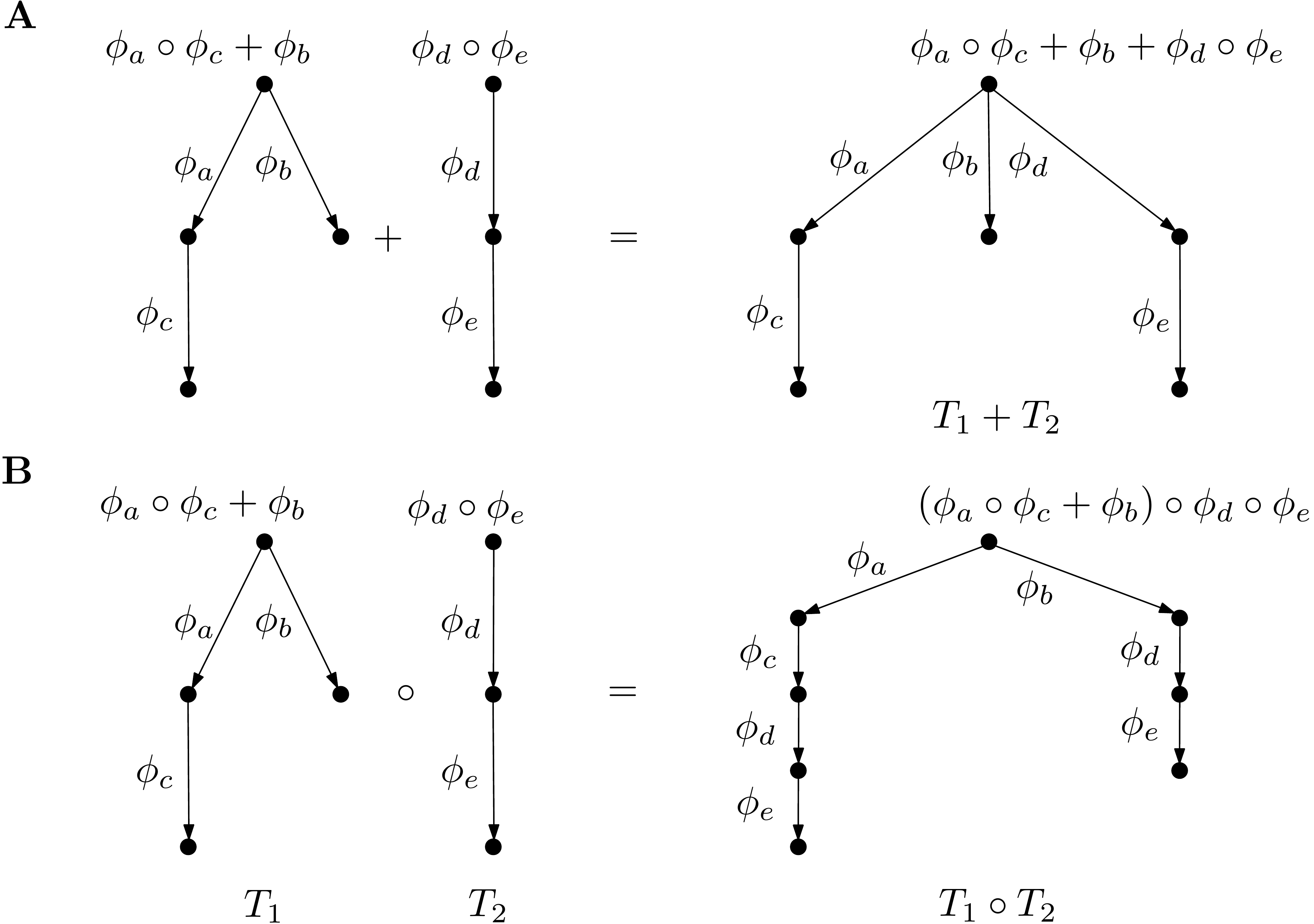}
  \caption{{\bf A} Addition of two trees $T_1$ and $T_2$: The roots of both trees are identified and the labels on all edges of both trees are retained.
  All vertex labels are given by the relation (\ref{eq:omega}).
  {\bf B} Multiplication of two trees $T_1$ and $T_2$: Each leaf of $T_1$ is replaced with a copy of $T_2$.
  Thereby the edge labels from the original trees are retained, yielding the vertex labels by relation (\ref{eq:omega}).
  Only the root labels are shown in the figure (To arrive at the form of the root label of $T_1 \circ T_2$ given here, the right-distributivity, cf. propery \ref{prop:distr}, is used).
  }
  \label{fig:tree_operations}
\end{figure}

\begin{rmk} \label{rmk:distributivity}
It follows directly from the definition of the addition and multiplication of trees that the operations are associative.
Moreover, the addition is commutative and the right distributivity
\begin{equation*}
    (T_1 + T_2) \circ T_3 = T_1 \circ T_3 + T_2 \circ T_3 
\end{equation*}
holds by construction.
\end{rmk}

The algebraic structure on $\mathfrak{T}(A)$ thus defined is compatible with the algebraic structure on $\mathcal{T}(\mathfrak{X})$ by mapping a tree $T \in \mathfrak{T}(A)$ to the label on its root
\begin{align}
    \label{eq:ev}
    \mathfrak{T}(A) &\xlongrightarrow{ev} \mathcal{T}(\mathfrak{X}) \\
    \nonumber
    (A,V,E,t,\omega_V,\omega_E) &\longmapsto \omega_V(t).
\end{align}

\begin{lem} \label{lem:tree_homo}
The map $ev: \mathfrak{T}(A) \longrightarrow \mathcal{T}(\mathfrak{X})$ is a homomorphism with respect to addition $+$ and multiplication $\circ$.
\end{lem}

\begin{proof}
The notation from Definition \ref{def:trees} is used.
Let $T_1,T_2 \in \mathfrak{T}(A)$ be two $A$-decorated rooted trees.

Let $T^+ = T_1 + T_2$.
By construction of $T^+$, the projection $\pi: V_1 \sqcup V_2 \rightarrow V^+$ is injective on all vertices except on the root.
Moreover, $\pi$ respects the level of a vertex, i.e. $\textrm{lv}(v) = \textrm{lv}(\pi(v))$, and the decoration function for vertices $v$ of level 1 satisfies
\begin{align*}
    \omega_V^+(v) =
    \begin{cases}
               \omega_{V1}(\pi^{-1}(v))  &\textrm{if }\pi^{-1}(v) \in V_1\\
               \omega_{V2}(\pi^{-1}(v))  &\textrm{if }\pi^{-1}(v) \in V_2.
    \end{cases}
\end{align*}
This yields the homomorphism property for addition
\begin{align*}
    \omega_V^+(t^+) &= \sum_{v \in \textrm{lv}_1(T^+)} \omega_E^+ ((t^+,v)) \circ \omega_V^+(v) \\
    &= \sum_{\substack{v \in \textrm{lv}_1(T^+), \\ \pi^{-1}(v) \in V_1}} \omega_{E1}((t_1,v)) \circ \omega_{V1}(\pi^{-1}(v)) + \\ 
    & +\sum_{\substack{v \in \textrm{lv}_1(T^+), \\ \pi^{-1}(v) \in V_2}}\omega_{E2}((t_2,v)) \circ \omega_{V2}(\pi^{-1}(v)) \\
    &= \omega_{V1}(t_1) + \omega_{V2}(t_2).
\end{align*}

Let $T^{\circ} = T_1 \circ T_2$.
By construction, $T_1$ is a subtree of $T^{\circ}$ and thus the respective vertices and edges of $T^{\circ}$ and $T_1$ can be identified.
It is now shown inductively that for all $v \in T_1$, considered as a subtree of $T^{\circ}$, the relation
\begin{equation} \label{eq:recursion}
    \omega_V^{\circ}(v) =  \omega_{V1}(v) \circ \omega_{V2}(t_2)
\end{equation}
\noindent holds.
For all leaves $l \in \textrm{lv}(T_1)$, the relation
\begin{equation*}
    \omega_V^{\circ}(l) = \omega_{V2}(t_2) = \textrm{id}\mid_{\mathfrak{X}} \circ \omega_{V2}(t_2) = \omega_{V1}(l) \circ \omega_{V2}(t_2)
\end{equation*}
holds by construction.
For the induction from vertex level $n$ (with $1 \leq n \leq \textrm{ht}(T_1)$ to $n-1$, let $v \in V_1 \setminus \textrm{lf}(T_1)$ be a vertex of level $n-1$.
The recursion (\ref{eq:omega}) yields
\begin{align*}
    \omega^{\circ}_V(v) &= \sum_{w \in \textrm{ch}(v) \cap V_1} \omega^{\circ}_E((v,w)) \circ \omega^{\circ}_V(w) \\
    &= \sum_{w \in \textrm{ch}(v) \cap V_1} \omega_{E1}((v,w)) \circ \omega_{V1}(w) \circ  \omega_{V2}(t_2) \\
    &=  \left[\sum_{w \in \textrm{ch}(v) \cap V_1} \omega_{E1}((v,w)) \circ \omega_{V1}(w) \right] \circ  \omega_{V2}(t_2) \\
    &= \omega_{V1}(v) \circ \omega_{V2}(t_2),
\end{align*}
where $\omega^{\circ}_E((v,w)) = \omega_{E1}((v,w))$ holds by definition, the second line is the induction hypothesis, and the third line follows from the right distributivity of the operations, cf. property \ref{prop:distr}.
In particular, the homomorphism property $\omega_V^{\circ}(t^+) =  \omega_{V1}(t_1) \circ \omega_{V2}(t_2)$ holds.
\end{proof}

Of particular importance are the trees decorated by the generating sets $\{ \phi_x \}_{x \in X} \cup \{0\}$ and $\{ \phi_r \}_{r \in R} \cup \{0\}$ of $\mathcal{S}^R$ and $\mathcal{S}$.
The respective tree algebras are denoted by
\begin{align*}
    \mathfrak{T} &:= \mathfrak{T}(\{ \phi_x \}_{x \in X} \cup \{0\}), \\
    \mathfrak{T}^R &:= \mathfrak{T}(\{ \phi_r \}_{r \in R} \cup \{0\}).
\end{align*}

\noindent There is a map with nice algebraic properties between the tree algebras
\begin{align*}
    \mathfrak{T} &\xlongrightarrow{\tau} \mathfrak{T}^R,
\end{align*}
\noindent which is defined based on the relation $\phi_x = \sum_{(x,r) \in C} \phi_r$ between the edge labels.
First, $\tau$ maps the trivial tree with one vertex in $\mathfrak{T}$ to the trivial tree in $\mathfrak{T}^R$.
Next, let $T_{\phi}$ be the decorated rooted tree with one edge which is labelled by $\phi$.
The tree $T_{\phi}$ is said to be the {\it atomic tree with label $\phi$}.
For an atomic tree $T_{\phi_x} \in \mathfrak{T}$, the label function $\phi_x$ can be uniquely decomposed as a sum of functions corresponding to reactions according to its definition, cf. Equation (\ref{eq:phi_x}):
\begin{equation*}
    \phi_x = \sum_{(x,r) \in C} \phi_r.
\end{equation*}
Thus, $\tau(T_{\phi_x})$ is defined as the sum of the corresponding atomic trees
\begin{equation*}
    \tau(T_{\phi_x}) := \sum_{(x,r) \in C} T_{\phi_r}.
\end{equation*}
A tree $T \in \mathfrak{T}$ of height one can be written as a finite sum of atomic trees, i.e. $T = \sum_{j=1}^m T_{\phi_{x_j}}$, and the map $\tau$ on $\mathfrak{T}^1$ is defined as
\begin{equation*} 
    \tau(T) := \sum_{j=1}^m \tau( T_{\phi_{x_j}}).
\end{equation*}
An arbitrary tree $T \in \mathfrak{T}^n$ of height $n$ can be written as $T = \sum_{j=1}^m  T_{\phi_{x_j}} \circ T_j $ for atomic trees $T_{\phi_{x_j}}$ and trees $T_j \in \mathfrak{T}^{\leq (n-1)}$ of height $\leq (n-1)$.
The map $\tau$ is defined recursively as
\begin{align} \label{eq:def_tau}
    \tau(T) := \sum_{j=1}^m \tau( T_{\phi_{x_j}}) \circ \tau(T_j).
\end{align}
The substitution process is illustrated in Fig. \ref{fig:Tr}A and an example of the construction $T \mapsto \tau(T)$ for the CRN in Fig. \ref{fig:Tr}B is given in Fig. \ref{fig:Tr}C.

\begin{figure}[ht]
  \centering
  \includegraphics[scale=0.25]{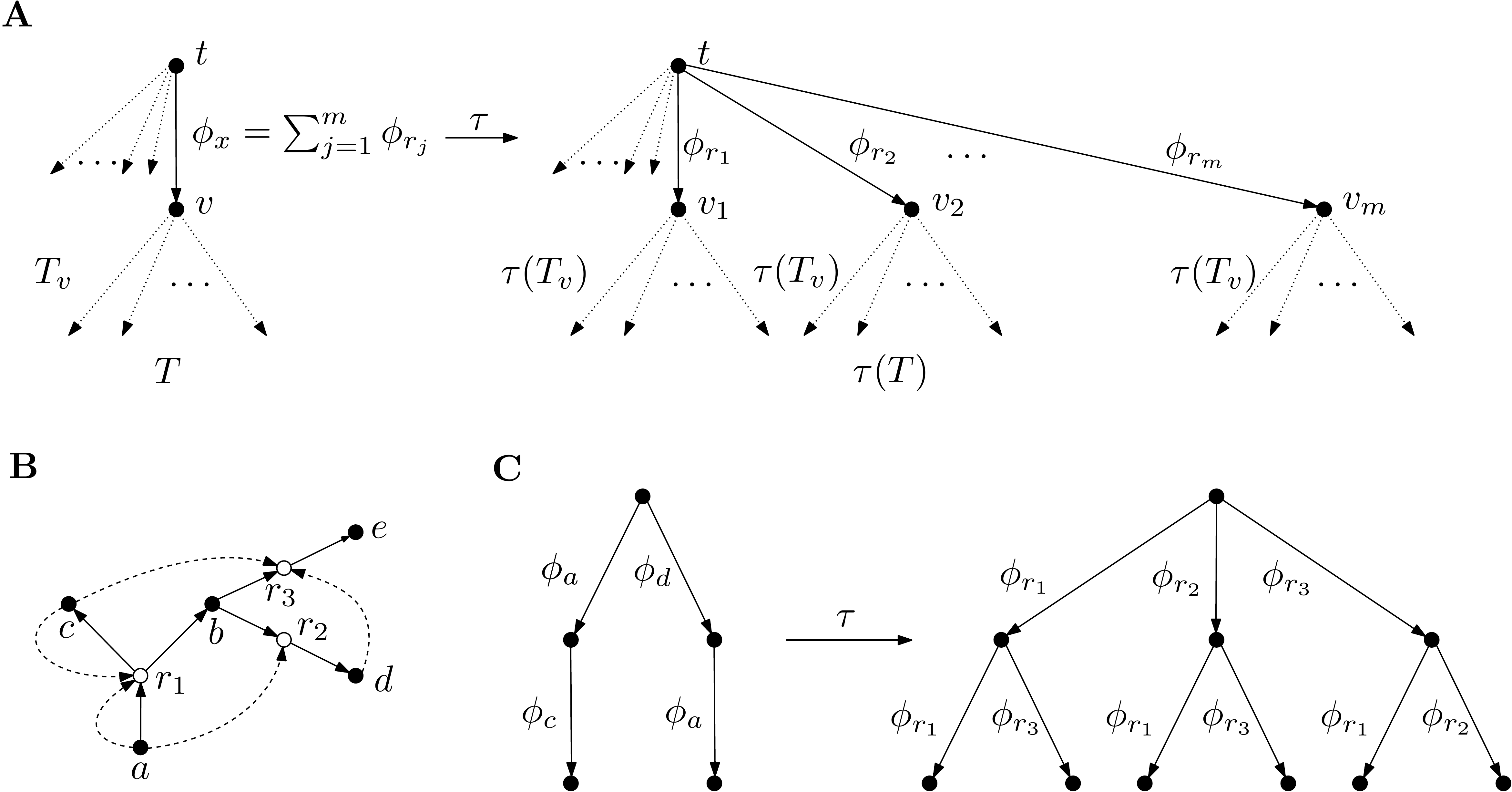}
  \caption{Illustration of the construction of the map $\tau: \mathfrak{T} \rightarrow \mathfrak{T}^R$.
  {\bf A} Illustration of the general procedure of replacing an edge of $T \in \mathfrak{T}$ with label $\phi_x$ by an edge for each summand in $\phi_x = \sum_{(x,r) \in C} \phi_r = \sum_{j=1}^m \phi_{r_j}$ with labels $\phi_{r_j}$.
  This is performed recursively starting with the terminal edges and working upwards toward the root.
  {\bf B} An example CRS with the functions of chemicals given by $\phi_a = \phi_{r_1} + \phi_{r_2}$,  $\phi_c = \phi_{r_1} + \phi_{r_3}$ and $\phi_d = \phi_{r_3}$.
  {\bf C} The map $\tau$ applied to the tree $T$ on the left for the CRS featured in B.
  }
  \label{fig:Tr}
\end{figure}

\begin{lem} \label{lem:tau_homo}
    The map $\tau: \mathfrak{T} \rightarrow \mathfrak{T}^R$ defined above is a homomorphism with respect to the addition and multiplication of trees.
    Moreover, the label of the root $\omega_V(t)$ is invariant under $\tau$ for any tree $T \in \mathfrak{T}$, i.e., using the evaluation map defined in (\ref{eq:ev}), the relation 
    \begin{align*}
    ev(T) = ev(\tau(T))    
    \end{align*}
    holds.
\end{lem}
\begin{proof}
    Let $T_1, T_2$ be two nontrivial trees in $\mathfrak{T}$.
    They can be written as $T_1 = \sum_{j=1}^m T_{\phi_{x_j}} \circ T_j$ and $T_2 = \sum_{j=m+1}^l T_{\phi_{x_j}} \circ T_j$ with atomic trees $T_{\phi_{x_j}}$.
    Their sum is given by $T_1 + T_2 = \sum_{j=1}^l T_{\phi_{x_j}} \circ T_j$ and the compatibility of $\tau$ with respect to addition follows from the definition (\ref{eq:def_tau}) and from the associativity of addition $\tau(T_1 + T_2) = \sum_{j=1}^l \tau(T_{\phi_{x_j}}) \circ \tau(T_j) = \tau(T_1) + \tau(T_2)$.
    
    The compatibility of $\tau$ with respect to multiplication is shown inductively on $\textrm{ht}(T_1)$.
    If $\textrm{ht}(T_1) = 1$, then $T_1 = \sum_{j=1}^m T_{\phi_{x_j}}$ and  $T_1 \circ T_2 = ( \sum_{j=1}^m T_{\phi_{x_j}} ) \circ T_2 = \sum_{j=1}^m ( T_{\phi_{x_j}} \circ T_2 )$ by Remark \ref{rmk:distributivity}.
    Thus, the property $\tau(T_1 \circ T_2) = \sum_{j=1}^m ( \tau(T_{\phi_{x_j}}) \circ \tau(T_2) ) = ( \sum_{j=1}^m \tau(T_{\phi_{x_j}})) \circ \tau(T_2) = \tau(T_1) \circ \tau(T_2)$
    follows from the definition (\ref{eq:def_tau}) and the right distributivity of the tree algebra $\mathfrak{T}^R$.
    Let $\textrm{ht}(T_1) = n$ and let $T_1$ be given by expression $T_1 = \sum_{j=1}^m T_{\phi_{x_j}} \circ T_j$ as above with $\textrm{ht}(T_j) \leq (n-1)$.
    Then $T_1 \circ T_2 =  ( \sum_{j=1}^m T_{\phi_{x_j}} \circ T_j ) \circ T_2 = \sum_{j=1}^m ( T_{\phi_{x_j}} \circ T_j  \circ T_2 )$ yields
    \begin{align*}
        \tau(T_1 \circ T_2) &= \sum_{j=1}^m \left[ \tau(T_{\phi_{x_j}}) \circ \tau(T_j  \circ T_2) \right] \\
        &= \sum_{j=1}^m \left[ \tau(T_{\phi_{x_j}}) \circ \tau(T_j)  \circ \tau(T_2) \right] \\
        &= \sum_{j=1}^m \left[ \tau(T_{\phi_{x_j}}) \circ \tau(T_j) \right] \circ \tau(T_2)  = \tau(T_1) \circ \tau(T_2),
    \end{align*}
    where the first line follows from the definition (\ref{eq:def_tau}), the second line from the induction hypothesis, and the third line from the right distributivity of the tree algebra.
    
    The invariance of the root label holds for an atomic tree as $ev(T_{\phi_x}) = \phi_x$ and $ev(\tau(T_{\phi_x})) = ev( \sum_{(x,r) \in C} T_{\phi_r}) = \sum_{(x,r) \in C} \phi_r = \phi_x$.
    It extends to the trees of height 1 by the associativity of the addition of trees and elements in $\mathcal{T}(\mathfrak{X})$.
    For a tree of arbitrary height this is verified inductively.
    Let $T = \sum_{j=1}^m T_{\phi_{x_j}} \circ T_j$.
    Its root label is determined by (\ref{eq:omega}) as $ev(T) = \sum_{j=1}^m \phi_{x_j} \circ ev(T_j)$.
    The root label of $\tau(T)$ is given by 
    \begin{equation*}
        ev(\tau(T)) = ev\left(\sum_{j=1}^m \tau(T_{\phi_{x_j}}) \circ \tau(T_j)\right) = \sum_{j=1}^m \phi_{x_j} \circ ev(\tau(T_j)),
    \end{equation*}
    which agrees with $ev(T)$ by induction hypothesis.
\end{proof}

These constructions yield the following central theorem.
\begin{thm} \label{thm:function}
With the maps $ev$ and $\tau$ defined above and the inclusion ${\iota:\mathcal{S} \cup \{ \textrm{id}\mid_{\mathfrak{X}}\}} \rightarrow  \mathcal{S}^R \cup \{ \textrm{id}\mid_{\mathfrak{X}} \}$, the following diagram commutes
\[\begin{tikzcd}
 \mathfrak{T} \arrow[r,twoheadrightarrow, "ev"]\arrow[d,hookrightarrow,"\tau"] & \mathcal{S} \cup \{ \textrm{id}\mid_{\mathfrak{X}} \} \arrow[d,hookrightarrow,"\iota"]\\
 \mathfrak{T}^R \arrow[r,twoheadrightarrow,"ev"]& \mathcal{S}^R \cup \{ \textrm{id}\mid_{\mathfrak{X}} \}
\end{tikzcd}\]
\noindent and all maps are homomorphisms.
Moreover, the evaluation maps $ev$ are surjective.
\end{thm}

\begin{proof}
 The homomorphism property follows from the Lemmata \ref{lem:tree_homo} and \ref{lem:tau_homo}.
 The commutativity of the diagram has also been proven in Lemma \ref{lem:tau_homo}.
 The evaluation maps are surjective because the generators $\{\phi_x\}_{x \in X} \cup \{0\}\subset \mathcal{S}$ and ${\{\phi_r\}_{x \in R} \cup \{0\} \subset \mathcal{S}^R}$ have preimages given by the atomic trees $\{T_{\phi_x}\}_{x \in X} \cup \{T_0\} \subset \mathfrak{T}$ and ${\{T_{\phi_r}\}_{r \in R} \cup \{T_0\} \subset \mathfrak{T}^R}$ combined with the fact that the tree algebras are closed under the operations of addition and multiplication.
\end{proof}

\noindent The finiteness of $\mathcal{S}$ yields the following corollary.
\begin{corr}
There is an $N$ such that the set of trees of height at most $\mathfrak{T}^{\leq N}$ maps surjectively onto $\mathcal{S} \cup \{ \textrm{id}\mid_{\mathfrak{X}} \}$.
\end{corr}

\noindent Moreover, Theorem \ref{thm:function} implies that the elements of $\mathcal{S}$ and $\mathcal{S}^R$ can be represented as decorated rooted trees by lifting the respective semigroup elements via the homomorphism $ev$.
\begin{defn} \label{def:minimalRepresentative}
A {\it tree representative} of an element $\phi \in \mathcal{S}$ is an element $T \in \mathfrak{T}$ such that $ev(T) = \phi$.
The representative is called {\it minimal} if it has no subtree $T'$ such that $ev(T') = \phi$.
The analogous definition holds for tree representatives of elements of $\mathcal{S}^R$ in $\mathfrak{T}^R$.
\end{defn}

\begin{rmk}[Biochemical interpretation of a tree] \label{rmk:StoSF2}
A tree $T \in \mathfrak{T}^R$ of level $n$ corresponds to a ''reaction mechanism'' of the network which can be described as follows:
The reactions at the terminal edges are carried out and their products are supplied to their tail vertices.
For each vertex, once it has received the products from all its outgoing edges, these products act as reactants for the reaction on its incoming edge.
This procedure is carried out iteratively for the levels of the tree and therefore takes $n-1$ steps for a tree of height $n$.
For a tree $T \in \mathfrak{T}$, the respective reaction mechanism is the reaction mechanism described $\tau(T) \in \mathfrak{T}^R$.
\end{rmk}

Finally, the decorated rooted trees in $\mathfrak{T}^R$ can be used to reformulate the F property given in Definition \ref{def:Fnetwork}.
In particular, the condition \ref{cond:F} can be encoded in a tree:
\begin{defn} \label{def:Ftree}
    Let $x \in X_F$ be a chemical for which the condition \ref{cond:F} holds.
    Let $R_1,\dotsc,R_n$ be the sets of reactions featured in \ref{cond:F} and denote by $T_{\phi_r}$ the atomic trees for the functions $\phi_r$.
    Define the trees $T_1^R,\dotsc,T_n^R \in \mathfrak{T}^R$ inductively as follows:
    Let
    \begin{equation} \label{eq:T1r}
        T_1^R := \sum_{r \in R_1} T_{\phi_r},
    \end{equation}
    and for $1<i<n$:
    \begin{equation} \label{eq:Tir}
        T_{i+1}^R := \sum_{r \in \bigcup_{j=1}^{i+1} R_j} T_{\phi_r} \circ T_i^R.
    \end{equation}
    The tree $T_n^R$ is said to be the {\it F-tree} for the element $x$.
    It is denoted by $T^R(x)$.
\end{defn}

\begin{lem} \label{lem:treesFproperty}
Let $x \in X_F$ be a chemical for which the condition \ref{cond:F} holds.
Then for the F-tree $T^R(x)$, the relation
\begin{equation*}
    x \in ev\left(T^R(x)\right)(\emptyset)
\end{equation*}
holds.
In other words, $T^R(x)$ represents a reaction mechanism that produces $x$ from the food set.
\end{lem}
\begin{proof}
Let $T_i^R$ be the trees from Definition \ref{def:Ftree} with $T_n^R =T^R(x)$.
It will be shown inductively that 
\begin{align} \label{eq:induction1}
    \bigcup_{j=1}^{i} \textrm{ran}(R_j) \subset ev\left(T_i^R\right)(\emptyset) \cup F
\end{align}
holds for all $i=1,\dotsc,n$ and therefore
\begin{align} \label{eq:induction2}
\bigcup_{j=1}^{i+1} \textrm{dom}(R_j) \subset ev\left(T_i^R\right)(\emptyset) \cup F    
\end{align}
holds for all $i=1,\dotsc,n-1$.
The inclusion (\ref{eq:induction2}) follows from (\ref{eq:induction1}) together with the conditions \ref{cond:F1} and \ref{cond:F2}.
Then, the claim $x \in ev\left(T_n^R\right)(\emptyset)$ will follow from the inclusion (\ref{eq:induction1}) together with the condition \ref{cond:F3}.

For $i=1$, the definition of $T_1^R$ gives
\begin{equation*}
    ev\left(T_1^R\right)(\emptyset) \cup F = \sum_{r \in R_1} \phi_r (\emptyset) \cup F.
\end{equation*}
From condition \ref{cond:F1}, i.e. $\textrm{dom}(R_1) \subset F$, it follows that $\textrm{ran}(R_1) = \sum_{r \in R_1} \phi_r(\emptyset) \subset ev\left(T_1^R\right)(\emptyset) \cup F$.
And from condition \ref{cond:F2}, i.e. $\textrm{dom}(R_2) \subset \textrm{ran}(R_1) \cup F$, together with \ref{cond:F1}, it follows that $\textrm{dom}(R_1) \cup \textrm{dom}(R_2) \subset ev\left(T_1^R\right)(\emptyset) \cup F$.

For $i+1$, one obtains
\begin{align*}
    ev\left(T_{i+1}^R\right)(\emptyset) \cup F &=  \sum_{r \in \bigcup_{j=1}^{i+1} R_j} \phi_r\left(ev\left(T_i^R\right)(\emptyset)\right) \cup F \\
        &\supset \bigcup_{j=1}^{i+1} \textrm{ran}(R_j), 
\end{align*}
where the final inclusion is obtained from the induction hypothesis $\bigcup_{j=1}^{i+1} \textrm{dom}(R_j) \subset ev\left(T_i^R\right)(\emptyset) \cup F$.
The conditions \ref{cond:F1} and \ref{cond:F2} imply now that $\bigcup_{j=1}^{i+2} \textrm{dom}(R_j) \subset ev\left(T^R_{i+1}\right)(\emptyset) \cup F $ for $i \leq n-2$.
\end{proof}

\section{Characterization of self-sustaining and self-generating CRS} \label{sec:RAF}

In this section, the representation of semigroup elements as trees is used to derive a succinct expression for the maximal RAF set of chemicals of a CRS as the fixed point of the generative dynamics $Y \mapsto \Phi_Y(\emptyset)$ with the initial condition $Y_0 = X_F$ (Theorem \ref{thm:maxRAF}).
In Section \ref{sec:prelim2}, it is shown that a CRS if RAF if and only if $\Phi_{X_F}(\emptyset) = X_F$ holds (Theorem \ref{thm:RAF}) and that the condition $\Phi_{X_F'}(\emptyset) = X_F'$ is sufficient for a set of chemicals $X_F' \subset X_F$ to be a RAF set of chemicals (Proposition \ref{prop:RAFsufficient}).
The latter statement is the key statement to prove that the fixed point of the dynamics, which is introduced in Section \ref{sec:dynamics}, satisfies the desired properties.
The importance of fixed points of the dynamics as functionally closed and therefore biologically relevant entities is discussed in Section \ref{sec:funcClosed}.

This whole section follows a logical structure which is analogous the structure of Section 4 in \cite{Loutchko2019}, where the analogous statements are proven for self-sustaining CRS.
However, the treatment of self-generating CRS is technically more involved, which is forced by the fact that the F property is more involved than the self-sustainment property of a CRS, cf. Remark \ref{rmk:selfSustaining}.\\
 
\noindent Throughout this section, fix a CRS  $(X,R,C,F)$ and let $\mathcal{S}$ be its semigroup model.

\subsection{Characterization of CRS with the RAF property} \label{sec:prelim2}

A CRS with the RAF property can be conveniently characterized via the set of chemicals generated by the maximal function of its semigroup model from the food set.

\begin{thm} \label{thm:RAF}
A CRS is RAF if and only if the maximal function $\Phi_{X_F}$  of its semigroup model satisfies
\begin{equation} \label{eq:XF}
\Phi_{X_F}(\emptyset) = X_F.
\end{equation}
\end{thm}

\begin{proof}
If the CRS is RAF, then by Lemma \ref{lem:treesFproperty}, the function $ev(T^R(x))$ satisfies $x \in ev(T^R(x))(\emptyset)$ for all $x \in X_F$.
This function is an element of $\mathcal{S}^R$ but not of $\mathcal{S}$ in general.
The RAF property allows to construct a tree $T(x) \in \mathfrak{T}$ such that $ev(T^R(x)) \leq ev(T(x))$ and thus $x \in ev(T(x))(\emptyset)$:
Choose a catalyst $y(r) \in X$ for each reaction $r \in R_i$ for all $R_i$ featured in the condition \ref{cond:F} for $x \in X_F$.
In analogy to the formulae (\ref{eq:T1r}) and (\ref{eq:Tir}), define
\begin{equation}
    T_1 := \sum_{r \in R_1} T_{\phi_{y(r)}},
\end{equation}
and for $1<i<n$:
\begin{equation}
    T_{i+1} := \sum_{r \in \bigcup_{j=1}^{i+1} R_j} T_{\phi_{y(r)}} \circ T_i
\end{equation}
with the atomic trees $T_{\phi_{y(r)}} \in \mathfrak{T}$ and set $T(x) := T_n$.
The properties \ref{eq:generators}, \ref{lemma:PO} and \ref{eq:order1} ensure that $ev(T^R(x)) \leq ev(T(x))$.
The function
\begin{align*}
    \Phi := \sum_{x \in X_F} ev(T(x))
\end{align*}
satisfies $X_F \subset \Phi(\emptyset)$ and thus the equality $\Phi(\emptyset) = X_F$ holds.
Therefore, $\Phi$ is the maximal function $\Phi_{X_F}$ of $\mathcal{S}$ and the claim $\Phi_{X_F}(\emptyset) = X_F$ holds.

To prove the reverse, assume that $\Phi_{X_F}(\emptyset) = X_F$ holds.
Choose a representative $T \in \mathfrak{T}$ for $\Phi_{X_F}$, i.e. a tree $T$ such that $ev(T)= \Phi_{X_F}$, and consider its image $\tau(T) \in \mathfrak{T}^R$.
Fix a chemical $x \in X_F$.
By Theorem \ref{thm:function}, the relation
\begin{align*}
    x \in ev(\tau(T))(\emptyset)
\end{align*}
holds.
Choose a subtree $T^{min}(x) \in \mathfrak{T}^R$ of $\tau(T)$ which is minimal under the condition
\begin{align} \label{eq:xInEv}
x \in ev(T^{min}(x))(\emptyset).
\end{align}
The existence of $T^{min}(x)$ follows from the existence of $\tau(T)$.
The sets $R_1,\dotsc,R_n$ featured in the condition \ref{cond:F} are constructed as follows:
Let the height of $T^{min}(x)$ be $n$ and define the set $R_i$ to contain the reaction corresponding to the labels of all edges whose heads have level $n+1-i$ for $1 \leq i \leq n$, i.e.
\begin{equation} \label{eq:RiFromTree}
    R_i := \{r \in R \textrm{ such that } \phi_r = \omega_E(e) \textrm{ for some } e \in \textrm{elv}_{n+1 -i}(T^{min}(x)) \},
\end{equation}
where $\omega_E$ is the decoration function for the edges of $T^{min}(x)$.
By the minimality of $T^{min}(x)$, the conditions \ref{cond:F1} and \ref{cond:F2} must be satisfied (reactions in any of the $R_i$ which do not satisfy the conditions could be omitted from the tree without violating the condition (\ref{eq:xInEv}) thus contradicting the minimality of $T^{min}(x)$).
The condition \ref{cond:F3} holds by construction of $T^{min}(x)$.
Finally, all reactions appearing as edge labels of $T^{min}(x)$, and thereby all reactions in the sets  $R_1,\dotsc,R_n$, are catalyzed because this holds for $\tau(T)$ by construction.
This concludes the proof.
\end{proof}

\begin{corr} \label{corr:maxRAF}
If $X_F' \subset X_F$ is a RAF set of chemicals, then the inclusion
\begin{equation*}
    X_F' \subset \Phi_{X_F'}(\emptyset)
\end{equation*}
holds.
\end{corr}

\begin{proof}
The maximal function $\Phi'_{X'_F}$ of the semigroup model $\mathcal{S}'(X',R\mid_{X'})$ satisfies $\Phi'_{X'_F}(\emptyset) = X'_F$ by Theorem \ref{thm:RAF}.
Its extension $ext(\Phi'_{X'_F})$, defined in Remark \ref{rmk:subsemigroup}, satisfies $ext(\Phi'_{X'_F})(\emptyset) = X'_F$ by definition.
The relation (\ref{eq:maxInequality}) gives
\begin{equation*}
    ext(\Phi'_{X'_F}) \leq \Phi_{X'_F}
\end{equation*}
and yields the claim when the functions above are applied to the empty set.
\end{proof}

For a RAF set of chemicals $X'_F$, the inclusion $X'_F \subset \Phi_{X'_F}(\emptyset)$ can be strict and therefore, in general, the equality $X'_F = \Phi_{X'_F}(\emptyset)$ is not satisfied.
However, it is a sufficient condition for $X'_F$ to be a RAF set of chemicals.

\begin{prop} \label{prop:RAFsufficient}
If the equality $X'_F = \Phi_{X'_F}(\emptyset)$ holds for a set of chemicals $X'_F \subset X_F$, then $X'_F$ is a RAF set of chemicals.
\end{prop}

\begin{proof}
The proof is analogous to the second part of the proof of Theorem \ref{thm:RAF}.
As in that proof, let $T \in \mathfrak{T}$ be a tree representative for the function $\Phi_{X'_F} \in \mathcal{S}$ of minimal height and let $T^{min}(x) \in \mathfrak{T}^R$ be a minimal subtree of $\tau(T)$ that satisfies $x \in ev(T^{min}(x))(\emptyset)$ for $x \in X'_F$ and has the same height as $T$.
Moreover, let $T$ be chosen such that all its edge labels are contained in the generating set $\{\phi_x\}_{x \in X'}$ of $\mathcal{S}(X'_F)$ (this is always possible since $T$ represents an element of $\mathcal{S}(X'_F)$).
This leads to the sets of reactions $R_1,\dots,R_n$ defined by (\ref{eq:RiFromTree}) and satisfying the condition \ref{cond:F} (the verification of this condition is analogous to the verification in the proof of Theorem \ref{thm:RAF}).
One only needs to ensure that all reactions $r$ contained in the $R_i$ satisfy $\textrm{supp}(r) \subset X'$, i.e. that they are elements of $R\mid_{X'}$, which is now shown:

The domain of each $R_i$ for $1 \leq i \leq n$ satisfies
\begin{equation*}
    \textrm{dom}(R_i) \subset \bigcup_{v \in \textrm{lv}_{n+1-i}(T^{min}(x))} \omega_V(v)(\emptyset) \cup F,
\end{equation*}
because the edges corresponding to the reactions with domains which are not contained in the set on the right hand side could be removed from $T^{min}(x)$, which would contradict its minimality (in the above formula, $\omega_V$ is the vertex decoration function of $T^{min}(x)$).
Therefore, it follows inductively that 
\begin{equation*}
    \textrm{ran}(R_i) \subset \bigcup_{v \in \textrm{lv}_{n-i}(T^{min}(x))} \omega_V(v)(\emptyset) \cup F.
\end{equation*}
Consider the functions
\begin{equation*}
    \phi_i^R:=\sum_{v \in \textrm{lv}_{n-i}(T^{min}(x))}\omega_V(v) \in \mathcal{S}^R,
\end{equation*}
which satisfy $\textrm{ran}(R_i) \subset \phi_i^R(\emptyset) \cup F$.
By construction of $T^{min}(x)$ as a subtree of $\tau(T)$ of the same height, the function $\phi_i^R$ is bounded from above by corresponding function $\phi_i$ constructed from $T$
\begin{equation*}
    \phi_i:=\sum_{v \in \textrm{lv}_{n-i}(T)}\omega_V(v) \in \mathcal{S}.
\end{equation*}
The $\phi_i$ are elements of $\mathcal{S}(X'_F)$ and are thus bounded from above by $\Phi_{X'_F}$.
This leads to the inclusion
\begin{align*}
    \textrm{ran}(R_i) \subset \phi_i^R(\emptyset) \cup F \subset \phi_i(\emptyset) \cup F \subset \Phi_{X'_F}(\emptyset) \cup F = X'.
\end{align*}
Together with the properites \ref{cond:F1} and \ref{cond:F2}, this yields $\textrm{supp}(R_i) \subset X'$ for all sets $R_i$.
\end{proof}

\noindent This proposition will be used to show that the fixed points of the dynamics, defined in the next section, are RAF sets of chemicals.

\subsection{Generative dynamics on a semigroup model and identification of the maximal RAF set of chemicals} \label{sec:dynamics}

The generative discrete dynamics of a CRS is introduced and used to determine the maximal RAF set of chemicals.
Starting with any set of chemicals $Y_0 \subset X_F$, there is a maximal function $\Phi_{Y_0}$ (Definition \ref{def:function}) that is supported on this set.
By acting on the empty set, $\Phi_{Y_0}(\emptyset)$ gives all non-food chemicals that can be generated from the food set by using functionality supported only on $Y_0$ and the food set.
The argument can be applied iteratively and gives rise to the following definition.

\begin{defn}
The {\it generative dynamics} of a CRS with the initial condition $Y_0 \subset X_F$ is generated by the propagator
\begin{align} \label{eq:propagatorg}
\mathcal{D}^g : \mathfrak{X} &\rightarrow \mathfrak{X}\\
\nonumber
Y &\mapsto \Phi_Y(\emptyset),
\end{align}
where $\Phi_Y$ is the function of $Y \subset X_F$.
The dynamics generated by $\mathcal{D}^g$ is parametrized by $\mathbb{Z}_{\geq 0}$ as
\begin{equation*} \label{eq:update}
Y_{n+1} = \mathcal{D}^g(Y_n) \text{ for all } n \in \mathbb{Z}_{\geq 0}.
\end{equation*}
\end{defn}

The generative dynamics has analogous properties to the sustaining dynamics and the reader is referred to Section 4.2. in \cite{Loutchko2019} for a more detailed discussion.
Here, only the properties needed for the proof of the main theorem are given.

\begin{rmk}
Due to the finiteness of the state space $\mathfrak{X}$, the dynamics either leads to a fixed point or to periodic behavior.
If the initial condition $Y_0$ leads to a fixed point, the dynamics is said to {\it stabilize} and the fixed point is denoted by $Y_0^{*g}$.
\end{rmk}

\begin{prop} \label{prop:fixedPoint}
Let the dynamics be given by $(Y_n)_{n \in \mathbb{Z}_{\geq 0}}$.
If $Y_1 \subset Y_0$, then
\begin{equation*}
Y_{n+1} \subset Y_n
\end{equation*}
holds for all $n \in \mathbb{Z}_{\geq 0}$ and the dynamics stabilizes.
The analogous statement holds for the case that $Y_1 \supset Y_0$.
\end{prop}

\begin{proof}
The proof proceeds by induction. 
By hypothesis $Y_1 \subset Y_0$ is satisfied.
Let $Y_n \subset Y_{n-1}$.
This implies the ordering of the respective functions $\Phi_{Y_n} \leq \Phi_{Y_{n-1}}$ by the property \ref{rmk:POtranssitivity}.
Together with the property \ref{eq:generators} this gives
\begin{equation*}
Y_{n+1} = \Phi_{Y_n}(\emptyset) \subset \Phi_{Y_{n-1}}(\emptyset) = Y_n.
\end{equation*}
The dynamics is thus given by the decreasing chain of subsets $Y_0 \supset Y_1 \supset ... \supset Y_n \supset Y_{n+1} \dots$ and, because $X_F$ is finite, the chain stabilizes.
The case $Y_1 \supset Y_0$ is treated analogously.
\end{proof}

\begin{lem} \label{lem:PRAFinclusion}
Let $X'_F \subset X_F$ be a RAF set of chemicals let $Y$ be a set that satisfies $X'_F \subset Y \subset X_F$.
Then the inclusion
 \begin{equation*}
     X'_F \subset \Phi_Y(\emptyset)
 \end{equation*}
holds.
\end{lem}
\begin{proof}
The chain of inclusions
\begin{equation*}
    X'_F \subset \Phi_{X'_F}(\emptyset) \subset \Phi_Y(\emptyset)
\end{equation*}
follows from the Corollary \ref{corr:maxRAF} and the property \ref{rmk:POtranssitivity}.
\end{proof}

\noindent Now the main theorem is stated and proven:

\begin{thm}[on the maximal RAF set of chemicals] \label{thm:maxRAF}
The maximal RAF set of chemicals of a CRS is the fixed point of the generative dynamics $(Y_n)_{n \in \mathbb{Z}_{\geq 0}}$ with the initial condition $Y_0 = X_F$, i.e. it is the set $X_F^{*g}$.
\end{thm}

\begin{proof}
It follows from Proposition \ref{prop:fixedPoint} that the dynamics has a fixed point $X_F^{*g}$.
By Proposition \ref{prop:RAFsufficient} this fixed point is a RAF set of chemicals.
It remains to show the maximality of $X_F^{*g}$:
For any RAF set of chemicals $X_F' \subset X_F$, the repeated application of Lemma \ref{lem:PRAFinclusion} implies that $X_F' \subset Y_n$ for all $n \in \mathbb{Z}_{\geq 0}$ and therefore $X_F' \subset X_F^{*g}$.
\end{proof}

\begin{corr} \label{corr:nilpotent}
A CRS with a nilpotent semigroup $(\mathcal{S},\circ)$ has no nontrivial RAF sets of chemicals.
\end{corr}

\begin{proof}
Let $X_F' \subset X_F$ be a nontrivial RAF set of chemicals.
Then $X_F' \subset \Phi_{X_F'}(\emptyset)$ holds by Corollary \ref{corr:maxRAF} and thus the condition \ref{eq:generators} implies that
\begin{equation*}
    X_F' \subset \Phi_{X_F'}^n(\emptyset)
\end{equation*}
for any power of $\Phi_{X_F'}$, i.e. $\Phi_{X_F'}^n$ is nonzero for any $n \in \mathbb{N}$.
\end{proof}

Nilpotent semigroups comprise the largest class of semigroups as any magma\footnote{A magma is a semigroup without the associativity property.} with the product of any three elements equal to zero is automatically a semigroup, cf. \cite{Satoh1994,Almeida1995}.
The above corollary weeds out all nilpotent semigroups as candidates for semigroup models of self-generating CRS and states that such models are located in a more interesting class of semigroups.

\begin{rmk}[Connection to the RAF algorithm]
\cite{Hordijk2004} have presented an algorithm to find the maximal RAF set of reactions.
It consists of a dynamics on the power set of reactions $\mathcal{P}(R)$ generated by $R' \mapsto \delta(\gamma(R'))$ with the initial condition $R_0 = R$.
The following two operations are performed iteratively:
\begin{enumerate} [label=(R\arabic*),leftmargin=1cm]
    \item For a set $R' \subset R$, remove all reactions from $R'$ that have no catalyst in $\textrm{supp}(R')$ until no further reductions can be made.
    This yields the set $\gamma(R')$.
    \label{alg:R1}
    \item For a set $R' \subset R$, until no further reductions can be made, remove all reactions $r$ from $R'$ that satisfy ${\textrm{dom}(r) \not \subset \Phi_{R'}(\emptyset) \cup F}$, where $\Phi_{R'}$ is the maximal function of the semigroup model $\mathcal{S}^R(R') := \langle \phi_r \rangle_{r \in R'}$.
    This yields the set $\delta(R')$.
    \label{alg:R2}
\end{enumerate}
Note that \ref{alg:R2} has been rephrased here to suit the language of semigroup models.
This is similar in spirit to the algorithm given in Theorem \ref{thm:maxRAF} by the generative dynamics $Y \mapsto \Phi_Y(\emptyset)$, where the sets of chemicals $Y$ should be thought of as the support of $R'$ featured in the RAF algorithm.
By forming the function $\Phi_Y$, all reactions without a catalyst in $Y = \textrm{supp}(R')$ are excluded, which corresponds to \ref{alg:R1}.
The application of the function $\Phi_Y$ to the empty set corresponds to the exclusion of all reactions without support in $\Phi_Y(\emptyset)$, i.e. to the step \ref{alg:R2}.
\end{rmk}

\subsection{Functionally closed RAF sets of chemicals} \label{sec:funcClosed}

In addition to the knowledge of the maximal RAF set of chemicals, the hierarchy of RAF subsets of chemicals plays an important role in the understanding of a CRS.
Of particular importance are the RAF sets of chemicals which satisfy the fixed point equation for the dynamics and are termed {\it functionally closed} RAF sets of chemicals in this section.
This is closely related to the notion of functionally closed sets of self-sustaining chemicals, which is developed in \cite{Loutchko2019}, Section 4.4.\\

If, for a RAF set of chemicals $X_F' \subset X_F$, the inclusion $X_F' \subset \Phi_{X_F'}(\emptyset)$ is strict, then the set is not stable in the sense that it will produce additional chemicals over time.
First, the chemicals in $Y_1 = \Phi_{X_F'}(\emptyset)$ will be generated from the food set, followed by chemicals in $Y_2 = \Phi_{Y_1}(\emptyset)$, etc.
By Proposition \ref{prop:fixedPoint}, this dynamics stabilizes at the fixed point $X_F'^{*g}$, which contains the original RAF set of chemicals $X_F'$.
Moreover, being a fixed point of the dynamics, $X_F'^{*g}$ satisfies
\begin{equation*} \label{eq:fixedPoint}
    \Phi_{X_F'^{*g}}(\emptyset)=X_F'^{*g}
\end{equation*}
and is thus a RAF set of chemicals by Proposition \ref{prop:fixedPoint}.
The set $X_F'^{*g}$ is not able to further catalyze the generation of chemicals outside of $X_F'^{*g}$ from the food set and is thus functionally closed.
This motivates the following definition.
\begin{defn}
The {\it functional closure} of a RAF set of chemicals $X_F' \subset X_F$ is the fixed point $X_F'^{*g}$ of the generative dynamics.
If $X_F'$ satisfies the fixed point equation $\Phi_{X_F'}(\emptyset)=X_F'$, then it is said to be a {\it functionally closed} RAF set of chemicals.
\end{defn}
Alternatively, the closure of a RAF set of chemicals can be characterized as follows:
\begin{lem} \label{lem:closureAlt}
The functional closure $X_F'^{*g}$ of a RAF set of chemicals $X_F'$ is the unique minimal functionally closed RAF set of chemicals which contains $X_F'$.
\end{lem}
\begin{proof}
Let $Y$ be a minimal functionally closed RAF set which contains $X_F'$ and let $(Y_n)_{n \in \mathbb{Z}_{\geq 0}}$ be the generative dynamics with the initial condition $Y_0 = X_F'$ and fixed point $X_F'^{*g}$.
Then $Y_n \subset Y$ holds for all $n \in \mathbb{Z}_{\geq 0}$, which can be verified by induction:
For $n=0$, the claim holds by assumption and the inductive step is verified by
\begin{equation*}
    Y_{n+1} = \Phi_{Y_n} (\emptyset) \subset \Phi_Y (\emptyset) = Y,
\end{equation*}
which follows from the property \ref{rmk:POtranssitivity}.
This implies that $X_F'^{*g} \subset Y$ and by the minimality of $Y$, the equality $X_F'^{*g} = Y$ must hold.
\end{proof}
\begin{rmk}
Note that the characterization of the closure of a RAF set of chemicals given by Lemma \ref{lem:closureAlt} does not extend to arbitrary sets, i.e. in general there does not exist a {\it unique} minimal functionally closed set of chemicals which contains $Y$ for a arbitrary set of chemicals $Y \subset X_F^{*g}$.
Fig. \ref{fig:notunique} provides an illustration.
The shown CRS is RAF and it has the functionally closed sets of chemicals given by $X_F = \{c,d,e\}$, $X_F' = \{c,d\}$ and $X_F'' = \{d,e\}$.
For the set $\{d\}$, there exists no unique minimal functionally closed set of chemicals which contains it.
\begin{figure}[ht]
  \centering
  \includegraphics[scale=0.25]{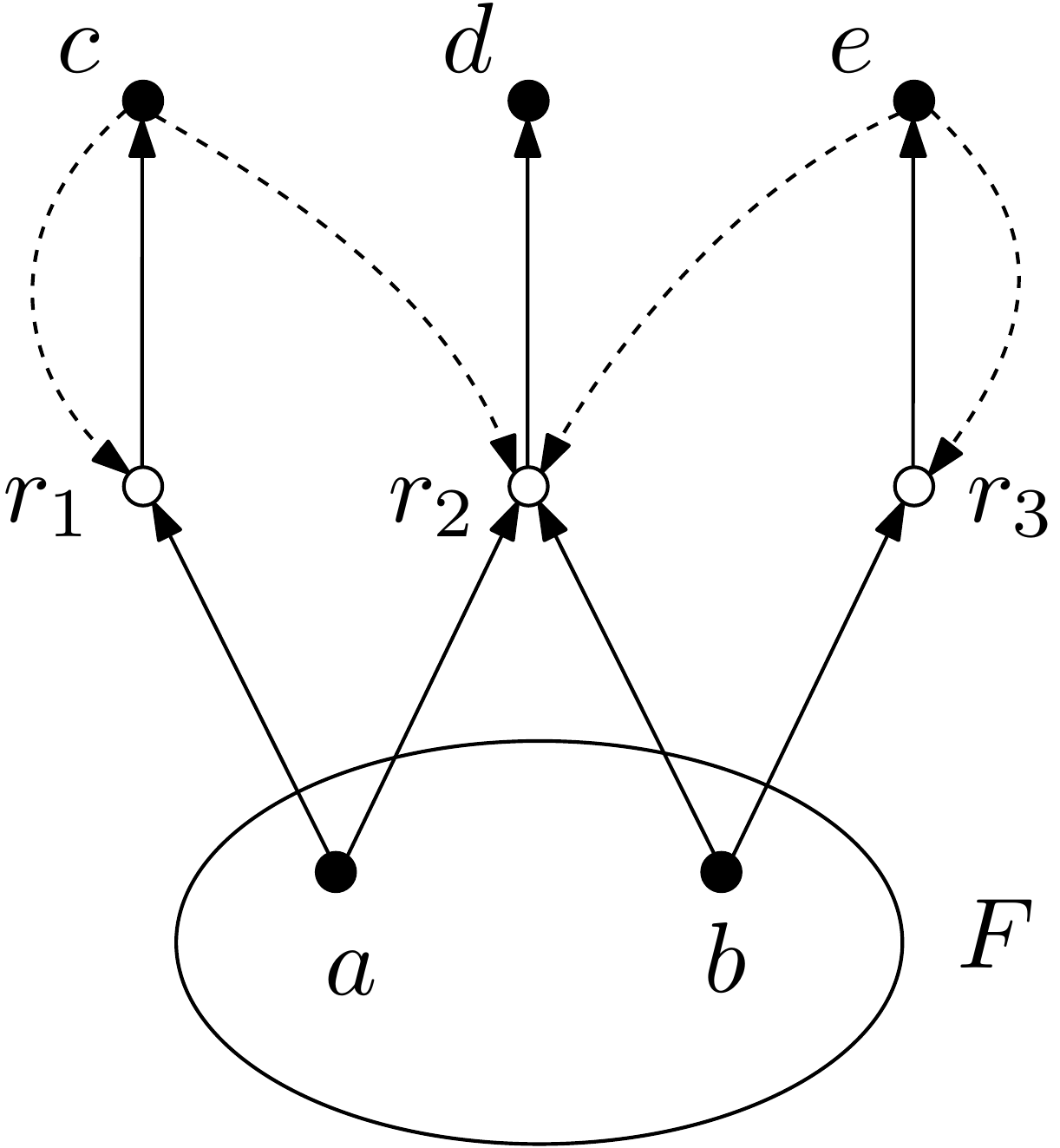}
  \caption{
  This CRS has the three functionally closed sets of chemicals given by $X_F = \{c,d,e\}$, $X_F' = \{c,d\}$ and $X_F'' = \{d,e,\}$.
  There is no unique minimal functionally closed set of chemicals which contains the set $\{d\}$.
  }
  \label{fig:notunique}
\end{figure}
\end{rmk}

The lattice of all functionally closed sets of chemicals can be obtained by the following construction.
Define the {\it reduced generative dynamics} by the propagator
\begin{align*} 
\mathcal{D}^{rg} : \mathfrak{X} &\rightarrow \mathfrak{X}\\
\nonumber
Y &\mapsto Y \cap \Phi_Y(\emptyset).
\end{align*}
This dynamics always stabilizes, and the fixed point for the initial condition $Y_0$ is denoted as $Y_0^{*rg}$.
The fixed point equation of this dynamics reads ${Y = Y \cap \Phi_Y(\emptyset)}$, which is equivalent to the fixed point equation $Y = \Phi_Y(\emptyset)$ of the generative dynamics.
For a set $Y \subset X_F$, define the set $\mathfrak{N}(Y) \subset \mathfrak{X}$ as
\begin{align*}
    \mathfrak{N}(Y) := \left\{ (Y \setminus \{y\})^{*rg} \textrm{ for }y \in Y \right\}.
\end{align*}
All of the sets contained in $\mathfrak{N}(Y)$ are functionally closed RAF sets of chemicals by Proposition \ref{prop:RAFsufficient}.
Moreover, let $X'_F \subset Y$ be a functionally closed RAF set of chemicals which is strictly contained in $Y$ and is maximal with this property.
Then there is a chemical $y \in Y \setminus X'_F$ and one verifies that $X'_F = (Y \setminus \{y\})^{*rg}$.

\noindent Now inductively define the following sets
\begin{align*}
    \mathfrak{N}^0 &:= \{ X_F^{*g} \} \\ 
    \mathfrak{N}^{i+1} &:= \bigcup_{Y \in  \mathfrak{N}^i} \mathfrak{N}(Y) \textrm{ for all }i \in \mathbb{Z}_{\geq 0}.
\end{align*}
Due to the finiteness of $\mathfrak{X}$, there is an $N \in \mathbb{N}$ such that $\mathfrak{N}^{i+1} = \{ \emptyset \}$ for all $i > N$.
The following theorem gives a description of the lattice of functionally closed RAF sets of chemicals of a CRS, which extends the characterization of the maximal RAF set of chemicals provided in Theorem \ref{thm:maxRAF}.

\begin{thm} \label{thm:funcClosed}
The set
\begin{equation*}
    \mathfrak{N} := \bigcup_{i=0}^N \mathfrak{N}^i
\end{equation*}
is the set of all functionally closed RAF sets of chemicals of the CRS.
\end{thm}
\begin{proof}
By construction, all elements of $\mathfrak{N}$ are functionally closed RAF sets of chemicals.
It remains to show that all functionally closed RAF sets of chemicals are indeed contained in $\mathfrak{N}$.
In this regard, recall that $\mathfrak{N}(Y)$ contains all maximal closed RAF sets of chemicals which are strictly contained in $Y$.
For a functionally closed RAF set of chemicals $X_F'$, there exists a chain of maximal length of functionally closed RAF sets of chemicals
\begin{equation*}
    X_F' = X_n \subsetneq X_{n-1} \subsetneq \dots \subsetneq X_0 = X_F^{*g}.
\end{equation*}
Then, $X_F'$ must be contained in $\mathfrak{N}^n$.
\end{proof}

This finishes the application of the semigroup models and their representations by decorated rooted tree to self-generating CRS.
The possible implications of the results of this section are now discussed.

\section{Discussion} \label{sec:discussion}

A general discussion of the semigroup models of CRS is given by \cite{Loutchko2019}, where, for example, algebraic properties and the possibility to analyze the computational properties of CRS with their semigroup models are expounded upon.\\

In this article, it was demonstrated how the language of semigroup models provides a natural framework to treat CRS with the RAF property, to determine the maximal RAF set of chemicals and to determine the lattice of functionally closed RAF sets of chemicals.
The technical basis is provided by the representation of the elements of the semigroup models as decorated rooted trees, because this representation is particularly useful in making the relation of semigroup elements with the F property precise.
It will be interesting to investigate whether such representations can be used more generally in the theoretical study of (not necessarily finite) semigroups and semirings.
Similar representations have turned out to be useful in the theory of self-similar groups introduced by \cite{nekrashevych2005}.\\

With regard to CRS theory, the notion of functionally closed sets of RAF chemicals is a very natural concept within the theory of semigroup models.
One is naturally led to consider the fixed points of the dynamics, which are RAF sets of chemicals by Proposition \ref{prop:fixedPoint}.
Moreover, Lemma \ref{lem:closureAlt} ensures that each RAF set of chemicals has a uniquely determined functional closure with nice properties.
The analysis of the lattice of functionally closed RAF sets of chemicals of a CRS within a living organism can potentially provide insights into the modular organization of its metabolism and the respective control mechanisms.
The fact that arbitrary subsets of $X_F$ - in contrast to RAF sets of chemicals - do not have a unique minimal functionally closed RAF set of chemicals which contains them, inspires further investigation of CRS of real biological systems.
If a chemical (or a set of chemicals) has a unique minimal functionally closed RAF set of chemicals to which it belongs, then one can conjecture that this chemical is specific for the respective functional module.
And it is likely that this chemical was acquired together with the respective module in the course of evolution.
If, however, this is not the case - such as for the chemical $d$ the example shown in Fig. \ref{fig:notunique}, then the respective chemical serves as a kind of mediator between the functional modules in which it is contained.
It will be interesting to test such hypotheses on CRS of biological systems and to develop new ones by applying the techniques provided by the semigroup formalism.\\

Another possibility suggested by the algebraic models of CRS is the coarse-graining obtained by taking quotients of the semigroups which are well-behaved with respect to the algebraic operations.
The technical difficulty is thereby to relate the quotients of functions, which live in $\mathcal{T}(\mathfrak{X})$ to quotients of the state space $\mathfrak{X}$ in a natural manner.
This work is currently being finalized.
This more algebraic approach provides an alternative way to reveal and analyze the modularity of a given CRS.
Whereas the lattice of functionally closed RAF sets of chemicals rely on the self-generating property, the quotient structures do not.
Therefore, in future, it will be interesting to compare the approach presented in Section \ref{sec:funcClosed} of this article to the algebraic coarse-graining procedures.\\

\backmatter

\bmhead{Acknowledgments}

I thank Tetsuya J. Kobayashi and all members of the Kobayashi lab for stimulating discussions.
This research is supported by JSPS KAKENHI Grant Numbers 19H05799 and 21K21308, and by JST CREST JPMJCR2011 and JPMJCR1927.


\bibliography{literatur}

\begin{thebibliography}{27}
\providecommand{\natexlab}[1]{#1}
\providecommand{\url}[1]{{#1}}
\providecommand{\urlprefix}{URL }
\providecommand{\doi}[1]{\url{https://doi.org/#1}}
\providecommand{\eprint}[2][]{\url{#2}}
 \bibcommenthead

\bibitem[{Almeida(1995)}]{Almeida1995}
Almeida J (1995) Finite semigroups and universal algebra, vol~3. World
  Scientific

\bibitem[{Dyson(1999)}]{Dyson1999}
Dyson F (1999) Origins of life. Cambridge University Press

\bibitem[{Eigen(1971)}]{Eigen1971}
Eigen M (1971) Selforganization of matter and the evolution of biological
  macromolecules. Naturwissenschaften 58(10):465--523

\bibitem[{G{\'a}nti(1975)}]{Ganti1975}
G{\'a}nti T (1975) Organization of chemical reactions into dividing and
  metabolizing units: the chemotons. BioSystems 7(1):15--21

\bibitem[{Gilbert(1986)}]{gilbert1986}
Gilbert W (1986) Origin of life: The {RNA} world. nature 319(6055):618--618

\bibitem[{Hordijk and Steel(2004)}]{Hordijk2004}
Hordijk W, Steel M (2004) Detecting autocatalytic, self-sustaining sets in
  chemical reaction systems. Journal of Theoretical Biology 227(4):451--461

\bibitem[{Hordijk and Steel(2017)}]{Hordijk2017}
Hordijk W, Steel M (2017) Chasing the tail: The emergence of autocatalytic
  networks. Biosystems 152:1--10

\bibitem[{Hordijk and Steel(2018)}]{Hordijk2018}
Hordijk W, Steel M (2018) Autocatalytic networks at the basis of life’s
  origin and organization. Life 8(4):62

\bibitem[{Hordijk et~al(2010)Hordijk, Hein, and Steel}]{Hordijk2010}
Hordijk W, Hein J, Steel M (2010) Autocatalytic sets and the origin of life.
  Entropy 12(7):1733--1742

\bibitem[{Hordijk et~al(2011)Hordijk, Kauffman, and Steel}]{Hordijk2011}
Hordijk W, Kauffman SA, Steel M (2011) Required levels of catalysis for
  emergence of autocatalytic sets in models of chemical reaction systems.
  International Journal of Molecular Sciences 12(5):3085--3101

\bibitem[{Hordijk et~al(2012)Hordijk, Steel, and Kauffman}]{Hordijk2012}
Hordijk W, Steel M, Kauffman S (2012) The structure of autocatalytic sets:
  Evolvability, enablement, and emergence. Acta Biotheoretica 60(4):379--392

\bibitem[{Hordijk et~al(2015)Hordijk, Smith, and Steel}]{Hordijk2015}
Hordijk W, Smith JI, Steel M (2015) Algorithms for detecting and analysing
  autocatalytic sets. Algorithms for Molecular Biology 10(1):15

\bibitem[{Joyce(1989)}]{joyce1989}
Joyce GF (1989) {RNA} evolution and the origins of life. Nature
  338(6212):217--224

\bibitem[{Joyce(2002)}]{joyce2002}
Joyce GF (2002) The antiquity of {RNA}-based evolution. nature
  418(6894):214--221

\bibitem[{Kauffman(1986)}]{Kauffman1986}
Kauffman SA (1986) Autocatalytic sets of proteins. Journal of Theoretical
  Biology 119(1):1--24

\bibitem[{Loutchko(2022)}]{Loutchko2019}
Loutchko D (2022) Semigroup models for biochemical reaction networks. arxiv
  preprint, arXiv:190804642

\bibitem[{Miller(1953)}]{miller1953}
Miller SL (1953) A production of amino acids under possible primitive earth
  conditions. Science 117(3046):528--529

\bibitem[{Nekrashevych(2005)}]{nekrashevych2005}
Nekrashevych V (2005) Self-similar groups. 117, American Mathematical Soc.

\bibitem[{Oparin(1957)}]{Oparin1957}
Oparin AI (1957) The Origin of Life on the Earth. Oliver \& Boyd, Edinburgh \&
  London

\bibitem[{Or{\'o}(1961)}]{oro1961}
Or{\'o} J (1961) Mechanism of synthesis of adenine from hydrogen cyanide under
  possible primitive earth conditions. Nature 191(4794):1193--1194

\bibitem[{Penny(2005)}]{penny2005}
Penny D (2005) An interpretive review of the origin of life research. Biology
  and Philosophy 20(4):633--671

\bibitem[{Rhodes and Nehaniv(2010)}]{Rhodes2010}
Rhodes J, Nehaniv CL (2010) Applications of automata theory and algebra: via
  the mathematical theory of complexity to biology, physics, psychology,
  philosophy, and games. World Scientific

\bibitem[{Rosen(1958)}]{Rosen1958}
Rosen R (1958) A relational theory of biological systems. Bulletin of
  Mathematical Biology 20(3):245--260

\bibitem[{Satoh et~al(1994)Satoh, Yama, and Tokizawa}]{Satoh1994}
Satoh S, Yama K, Tokizawa M (1994) Semigroups of order 8. In: Semigroup Forum,
  Springer, pp 7--29

\bibitem[{Sousa et~al(2015)Sousa, Hordijk, Steel, and Martin}]{Sousa2015}
Sousa FL, Hordijk W, Steel M, et~al (2015) Autocatalytic sets in e. coli
  metabolism. Journal of Systems Chemistry 6(1):4

\bibitem[{Steel(2000)}]{Steel2000}
Steel M (2000) The emergence of a self-catalysing structure in abstract
  origin-of-life models. Applied Mathematics Letters 13(3):91--95

\bibitem[{Varela et~al(1974)Varela, Maturana, and Uribe}]{Varela1974}
Varela FG, Maturana HR, Uribe R (1974) Autopoiesis: the organization of living
  systems, its characterization and a model. BioSystems 5(4):187--196

\end{thebibliography}


\end{document}